\documentclass[12pt,letter]{amsart}
\usepackage[latin1]{inputenc}
\usepackage{tikz}
\usetikzlibrary{arrows.meta}
\usetikzlibrary{bending}
\usepackage{verbatim}
\usepackage{xcolor}
\usepackage[all,color]{xy}
\usepackage{bold-extra,tensor,mathrsfs,import,stmaryrd,wasysym,amsmath,amssymb,amsthm,amscd,amsxtra,amsfonts,mathrsfs,graphicx,enumerate,bm,slashed,array,setspace,enumitem,minibox,booktabs,rotating,multirow,adjustbox,hhline,arydshln,placeins}
\usepackage{import,stmaryrd,wasysym,amsmath,amssymb,amsthm,amscd,amsxtra,amsfonts,mathrsfs,graphicx,enumerate,bm,slashed,array,setspace,enumitem,minibox,booktabs,rotating,multirow,adjustbox,hhline,arydshln,placeins}
\usepackage{tikz-feynman}
\input xy
\xyoption{all}
\usetikzlibrary{shapes,arrows}
\usetikzlibrary{positioning}
\usetikzlibrary{calc}
\input xy
\xyoption{all}
\newtheorem{Theorem}{Theorem}[section]
\newtheorem{Lemma}[Theorem]{Lemma}
\newtheorem{Proposition}[Theorem]{Proposition}
\newtheorem{Corollary}[Theorem]{Corollary}

\newtheorem{Remark}[Theorem]{Remark}

\newtheorem{Definition}[Theorem]{Definition}

\newtheorem*{Definition*}{Definition}
\DeclareMathOperator{\im}{im}
\DeclareMathOperator{\vol}{vol}

\advance\evensidemargin-.5in
\advance\oddsidemargin-.5in
\advance\textwidth1in

\definecolor{darkgreen}{rgb}{0.0, 0.5, 0.0}
\definecolor{darkred}{rgb}{0.7, 0.11, 0.11}

\makeatletter
\DeclareRobustCommand{\cev}[1]{%
  {\mathpalette\do@cev{#1}}%
}
\newcommand{\do@cev}[2]{%
  \vbox{\offinterlineskip
    \sbox\z@{$\m@th#1 x$}%
    \ialign{##\cr
      \hidewidth\reflectbox{$\m@th#1\vec{}\mkern4mu$}\hidewidth\cr
      \noalign{\kern-\ht\z@}
      $\m@th#1#2$\cr
    }%
  }%
}
\makeatother

\begin{document}

\author{Charlie Beil}
\address{Institut f\"ur Mathematik und Wissenschaftliches Rechnen, Universit\"at Graz, Heinrichstrasse 36, 8010 Graz, Austria.}
 \email{charles.beil@uni-graz.at}
\title[A derivation of the standard model particles from internal spacetime]{A derivation of the standard model particles\\ from the Dirac Lagrangian on internal spacetime} 

\begin{abstract}
`Internal spacetime' is a modification of general relativity that was recently introduced as an approximate spacetime geometric model of quantum nonlocality.
In an internal spacetime, time is stationary along the worldlines of fundamental (dust) particles. 
Consequently, the dimensions of tangent spaces at different points of spacetime vary, and spin wavefunction collapse is modeled by the projection from one tangent space to another.  
In this article we develop spinors on an internal spacetime, and construct a new Dirac-like Lagrangian $\mathcal{L} = \bar{\psi}(i \slashed \partial - \hat{\omega}) \psi$ whose equations of motion describe their couplings and interactions. 
Furthermore, we show that hidden within $\mathcal{L}$ is the entire standard model: $\mathcal{L}$ contains precisely three generations of quarks and leptons, the electroweak gauge bosons, the Higgs boson, and one new massive spin-$2$ boson; gluons are considered in a companion article.
Specifically, we are able to derive the correct spin, electric charge, and color charge of each standard model particle, as well as predict the existence of a new boson.
\end{abstract}

\maketitle

\tableofcontents

\section{Introduction} \label{intro}

The aim of internal spacetime geometry, recently introduced in \cite{B6}, is to model quantum phenomena with spacetime metrics that are degenerate. 
The geometry is obtained by replacing the worldlines of fundamental dust particles by single points, but \textit{without contracting the worldlines to points} as is done when taking geometric or topological quotients.
Time then remains stationary for free fundamental dust particles.
The framework for this model was introduced in the study of nonnoetherian coordinate rings in algebraic geometry \cite{B3,B4}, which in turn arose from the study of non-superconformal brane tilings in string theory \cite{HK, FHVWK,B5}.

We begin with the basic construction of the geometry.

\begin{Definition} \rm{
Let $(\tilde{M},g)$ be an orientable Lorentzian $4$-manifold.
Consider a (locally finite) set of dust particles on $\tilde{M}$ with worldlines $\beta_i \subset \tilde{M}$.
We call the set
\begin{equation*}
M := (\tilde{M} \setminus (\cup_i \beta_i)) \cup (\cup_j \{\beta_j \}),
\end{equation*}
where each $\beta_j$ is a single point of $M$, 
an \textit{internal spacetime}, or simply \textit{spacetime}. 
We call $\tilde{M}$ the \textit{external spacetime} of $M$, and the dust particles \textit{pointons}.
}\end{Definition}

We want to define an `internal metric' for which the worldline $\beta$ of a pointon, although a continuum of distinct 0-dimensional points in $\tilde{M}$, may be viewed as a single `1-dimensional point' in $M$; that is, we do not want to `throw away' $\tilde{M}$ in constructing $M$.

Since $\beta$ is a single point of $M$, a tangent vector, or $4$-velocity, cannot be defined along $\beta$ in $M$.
Therefore, to construct an internal metric at a point $p \in \tilde{M}$ from the external metric $g_{ab}$, it must project out each vector $v$ tangent to a (geodesic) pointon worldline $\beta \subset \tilde{M}$ at $p$. 
Recall the orthogonal projection of a timelike unit vector $v$:
\begin{equation*}
[v]_{ab} := g_{ab} - v_av_b. 
\end{equation*}
The case where $v$ is null is more involved, and is given in \cite[Section 5]{B6}.

\begin{Definition} \rm{
Fix a point $p \in \tilde{M}$.
We call the metric $g_{ab}: \tilde{M}_p \otimes \tilde{M}_p \to \mathbb{R}$ an \textit{external metric} at $p$.
Let $v_1, \ldots, v_n \in \tilde{M}_p$ be the tangent vectors to the pointon worldlines $\beta_1, \ldots, \beta_n$ at $p$.
The corresponding \textit{internal metric} is the degenerate symmetric rank-$2$ tensor given by the composition of projections
\begin{equation} \label{im}
h = h_p = \tensor{h}{^a_b} := \tensor{[v_1]}{^a_c}\tensor{[v_2]}{^c_d} \cdots \tensor{[v_n]}{^e_b} : \tilde{M}^*_p \otimes \tilde{M}_p \to \mathbb{R}.
\end{equation}
The \textit{(internal) tangent space} at $p$ is then the image of $h$ at $p$,
\begin{equation*}
M_p := \im h = \{ v^a \in \tilde{M}_p \, | \, \tensor{h}{^a_b}\tensor{v}{^b} = \tensor{v}{^a} \} \subseteq \tilde{M}_p.
\end{equation*}
}\end{Definition}

\begin{Remark} \rm{
A timelike geodesic $\beta \subset \tilde{M}$ is parameterized by its proper time. 
A question, then, is whether this proper time is lost in constructing $M$, since in $M$ all the interior points of a pointon worldline $\beta$ are identified.
The proper time does indeed disappear for particles within $\beta$, but the embedding image of $\beta$ in $\tilde{M}$ remains intact in $M$ under the internal metric $h$, that is, $\beta$ is not contracted to a $0$-dimensional point.
Thus, the proper time continues to parameterize $\beta$ to \textit{observers outside of $\beta$}, and this is why we call $\beta$ a $1$-dimensional point of $M$.
 }\end{Remark}

Recall that an orientation of a vector space $V$ is given by fixing an ordered basis $\mathcal{B}$ of $V$, and declaring any ordered basis to be positive (resp.\ negative) if it can be obtained from $\mathcal{B}$ by a base change with a positive (resp.\ negative) determinant.
As we shall find, subspace orientation plays an essential role in internal spacetime geometry.

Since the $4$-velocity $v \in \tilde{M}_{\beta(t)}$ of a pointon vanishes on $M$, $h(v) = \tensor{h}{^a_b}v^b = 0$, we want to replace $v$ with a new geometric object that is intrinsic to spacetime $M$ and independent of external spacetime $\tilde{M}$.

\begin{Definition} \rm{
The \textit{internal $4$-velocity} of a pointon with worldline $\beta \subset \tilde{M}$ and $4$-velocity $v$ is the pseudo-form
\begin{equation*}
\breve{v}_{a \cdots b} := o_{\operatorname{ker}h} \star \! \vol (\ker h) \in {\bigwedge} \! ^{\dim M_{\beta(t)}} \, M_{\beta(t)}^*,
\end{equation*}
where $o_{\operatorname{ker}h} \in \{ \pm 1 \}$ is a free parameter independent of any orientation of $\tilde{M}_{\beta(t)}$, and $\star \! \vol (\ker h)$ is the Hodge dual of the volume form of the kernel of $h$.
Note that the rank of $\breve{v}$ changes along $\beta \subset \tilde{M}$ whenever the dimension of the tangent space $M_{\beta(t)}$ changes.
}\end{Definition}

Let $\beta \subset \tilde{M}$ be a timelike pointon worldline with $4$-velocity $v$, and let $e_0, \ldots, e_3$ be an orthonormal tetrad along $\beta$ for which $e_0 = v$.
Fix $p = \beta(t) \in \tilde{M}$.

$\bullet$ If $\dim M_p = 3$, then the internal $4$-velocity at $p$ is
\begin{equation*}
\breve{v}^{abc} = o_0 \, e_1 \wedge e_2 \wedge e_3,
\end{equation*}
where $o_0 \in \{ \pm 1 \}$ is a free choice of time orientation (in the rest frame of the pointon), independent of any orientation of $\tilde{M}_p$. 
\textit{We identify $o_0$ with the electric charge of the pointon.}
Although this is similar to the St\"uckelberg-Feynman interpretation of antimatter in quantum field theory \cite{St}, time does not flow along $\beta$: time does not flow backwards along $\beta$ just as it does not flow forwards, since $\beta$ is a single point of spacetime $M$.

$\bullet$ If $\dim M_p = 1$ and $\ker h$ is spanned by, say, $e_0, e_1, e_2$, then 
\begin{equation} \label{spin vector s}
\breve{v}^a = o_0o_{12} \, e_3,
\end{equation}
where $o_{12} \in \{ \pm 1 \}$ is a free choice of orientation of the plane spanned by $e_1$ and $e_2$.
We call $\breve{v}^a$ the \textit{spin vector} of the pointon at $p$, and identify $o_{12}$ as spin, up $\uparrow$ or down $\downarrow$, in the $e_3$ direction.
The vector $s := \breve{v}^a$ is then parallel transported along $\beta$ until it is projected under $h$ onto the sequent $1$-dimensional tangent space $M_{\beta(t')}$.

Consider a pointon with spin vector $s = \breve{v}^a$ and worldline $\beta \subset \tilde{M}$ such that the dimension of the tangent space at $\beta(0)$ is a local minimum along $\beta$.
If $|h(s)| \not = 0$, denote by $\hat{h}(s) := h(s)/|h(s)|$ the normalization of $h(s)$, and denote by $M_{p \to q}$ the parallel transport of $M_p$ to $q$ along $\beta$.
We have the following:

\begin{itemize}
 \item[\textsc{(a)}] As $s$ enters a lower dimensional internal space at $\beta(0)$, it is projected onto $M_{\beta(0)}$ by the internal metric $h: \tilde{M}_{\beta(0)} \to M_{\beta(0)}$.
\item[\textsc{(b)}] As $h(s)$ \textit{exits} a lower dimensional internal space at $\beta(0)$, the time reversal of (\textsc{a}) occurs: a unit vector $s' \in M_{\beta(\epsilon) \to \beta(0)} \subset \tilde{M}_{\beta(0)}$ is chosen for which
\begin{equation} \label{equals}
h(s') \left| h(s) \right| = h(s) \left| h(s') \right|, \ \ \ \ \text{ or equivalently, } \ \ \ \ 
h(s) \! \cdot \! s' \geq 0.
\end{equation}
This simplifies to $\hat{h}(s) = \hat{h}(s')$ whenever $h(s)$ and $h(s')$ are nonzero; and if $h(s) = 0$, then (\ref{equals}) implies that $s'$ is unconstrained.
If $h(s) \not = 0$, then the probability that $s'$ is chosen is given by the \textit{Kochen-Specker probability}: 
\begin{equation*}
p(s' | h(s)) = \tfrac{1}{\pi} \hat{h}(s) \! \cdot \! s'.
\end{equation*}
\end{itemize}

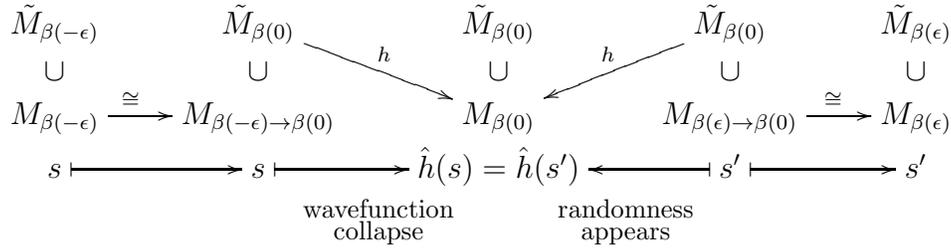
\begin{equation*}
\xymatrix@R-2pc{
\tilde{M}_{\beta(-\epsilon)} & \tilde{M}_{\beta(0)} \ar^{h}[rdd]
& \tilde{M}_{\beta(0)} & \tilde{M}_{\beta(0)}  \ar_{h}[ldd]
& \tilde{M}_{\beta(\epsilon)}\\
\cup & \cup & \cup & \cup & \cup \\
M_{\beta(-\epsilon)} \ar^{\cong \ \ \ \ \ }[r] & M_{\beta(-\epsilon) \to \beta(0)} & M_{\beta(0)}
& M_{\beta(\epsilon) \to \beta(0)} \ar^{\ \ \cong}[r] & M_{\beta(\epsilon)} \\
s \ar@{|->}[r] & s \ar@{|->} [r] 
\ar@{}_{\substack{ \ \\ \ \\ \text{\footnotesize{wavefunction}} \\ \text{\footnotesize{collapse}}}} [r] &
\hat{h}(s) = \hat{h}(s') \ar@{<-|} [r] 
\ar@{}_{\substack{ \ \\ \ \\ \text{ \ \ \footnotesize{randomness}} \\ 
\text{ \ \ \footnotesize{\textcolor{white}{l} appears \textcolor{white}{l}}}}} [r] & s' \ar@{|->}[r] & s' 
} 
\end{equation*}

Suppose $s$ exits a $1$-dimensional tangent space $M_p$ at $p \in \tilde{M}$, is parallel transported along $\beta$ for some time, and then enters another $1$-dimensional tangent space $M_q$ at $q$.
We then say the spin $s$ of the pointon is \textit{prepared} at $p$ and \textit{measured} at $q$.
We thus obtain a spacetime geometric realization of the Kochen-Specker model of spin \cite{KS}.
This implies, in particular, that the Born rule for spin holds in our model.

\begin{table}
\caption{A composite model of the standard model particles with precisely one new massive spin-$2$ boson, denoted $x$, derived from the Dirac equation. 
In short, a `geom' is a mass term of the Dirac Lagrangian. 
Subscripts denote spin states.
Here, ${a,b \in \{\uparrow, \downarrow \}}$.}
\label{table1}
\begin{center}
    \begin{tabular}{l||rrrr}
elec.\ charge & \multicolumn{4}{l}{all possible geoms $[\psi^-_1 \psi^+_1, \psi^-_2 \psi^+_2] = [\psi^-_2 \psi^+_2, \psi^-_1 \psi^+_1]$}\\
\hhline{=====}
$0$ & & $\gamma_{\uparrow} = [ \uparrow \downarrow, 00]$ & $Z_{\uparrow} = [ \uparrow \uparrow, 00]$ & $x_{ab} = [a 0, 0 b]$\\
(bosonic) & & $\gamma_{\downarrow} = [ \downarrow \uparrow, 00]$ & $Z_{\downarrow} = [ \downarrow \downarrow, 00]$ & $ x_0 = [ \downarrow \! {*}, {*} \! \downarrow]$ \\
& & $Z_0 = [ \uparrow \downarrow, \downarrow \uparrow]$ & $H = [\uparrow \uparrow, \downarrow \downarrow]$ & \\
\hdashline
$0$ & $\nu_e = [00, {*} \! \downarrow]$ & $\nu_{\mu} = [\downarrow \uparrow, {*} \! \downarrow]$ & $\nu_{\tau} = [\uparrow \uparrow, {*} \! \downarrow]$ & \\
(fermionic) & $\bar{\nu}_e = [00, \downarrow \! {*}]$ & $\bar{\nu}_{\mu} = [\uparrow \downarrow, \downarrow \! {*}]$ & $\bar{\nu}_{\tau} = [\uparrow \uparrow, \downarrow \! {*}]$ & \\
\hdashline
$-1$ &$e_{\uparrow} = [00, \uparrow \! 0]$ & $\mu_{\uparrow} = [\downarrow \uparrow, \uparrow \! 0]$ & $\tau_{\uparrow} = [ \downarrow \downarrow, \uparrow \! 0]$ & $\text{\footnotesize{$W^-_{\uparrow}$}} = [{*} \! \downarrow, \uparrow \! 0]$ \\
& $e_{\downarrow} = [00, \downarrow \! 0]$ & $\mu_{\downarrow} = [ \uparrow \downarrow, \downarrow \! 0]$ & $\tau_{\downarrow} = [\uparrow \uparrow, \downarrow \! 0]$ & $\text{\footnotesize{$W_{\downarrow}^-$}} = [{*} \! \downarrow, \downarrow \! 0]$ \\
 & & & & $\text{\footnotesize{$W_{0}^-$}} = [\downarrow \! {*}, \uparrow \! 0]$ \\
\hdashline
$+1$ & $\bar{e}_{\downarrow} = [00, 0 \! \uparrow ]$ & $\bar{\mu}_{\downarrow} = [\uparrow \downarrow, 0 \! \uparrow ]$ & $\bar{\tau}_{\downarrow} = [\downarrow \downarrow, 0 \! \uparrow]$ & $\text{\footnotesize{$W_{\downarrow}^+$}} = [\downarrow \! {*}, 0 \! \uparrow]$ \\
& $\bar{e}_{\uparrow} = [00, 0 \! \downarrow ]$ & $\bar{\mu}_{\uparrow} = [\downarrow \uparrow, 0 \! \downarrow]$ & $\bar{\tau}_{\uparrow} = [\uparrow \uparrow, 0 \! \downarrow]$ &
$\text{\footnotesize{$W_{\uparrow}^+$}} = [\downarrow \! {*}, 0 \! \downarrow]$\\
 & & & & $\text{\footnotesize{$W_{0}^+$}} = [{*} \! \downarrow, 0 \! \uparrow]$ \\
\hline
$\tfrac{-1}{3} + 1 = \tfrac{2}{3}$ & $u_{\uparrow} = \textcolor{red}{\pmb{(}} 00, {*} \! \downarrow]$ & $c_{\uparrow} = \textcolor{red}{\pmb{(}} \! \! \downarrow \uparrow, {*} \! \downarrow]$ 
& $t_{\uparrow} = \textcolor{red}{\pmb{(}} \! \! \uparrow \uparrow, {*} \! \downarrow]$ & \\
& $u_{\downarrow} = \textcolor{red}{\pmb{(}} 00, {*} \! \uparrow]$ & $c_{\downarrow} = \textcolor{red}{\pmb{(}} \! \! \uparrow \downarrow, {*} \! \uparrow]$ 
& $t_{\downarrow} = \textcolor{red}{\pmb{(}} \! \! \downarrow \downarrow, {*} \! \uparrow]$ &
\\
\hdashline
$\tfrac{1}{3} - 1 = \tfrac{-2}{3}$ & $\bar{u}_{\downarrow} = [00, \downarrow \! {*} \textcolor{red}{\pmb{)}}$ & $\bar{c}_{\downarrow} = [\uparrow \downarrow, \downarrow \! {*} \textcolor{red}{\pmb{)}}$ & $\bar{t}_{\downarrow} = [\uparrow \uparrow, \downarrow \! {*} \textcolor{red}{\pmb{)}}$ &
\\
 & $\bar{u}_{\uparrow} = [00, \uparrow \! {*} \textcolor{red}{\pmb{)}}$ & $\bar{c}_{\uparrow} = [\downarrow \uparrow, \uparrow \! {*} \textcolor{red}{\pmb{)}}$ & $\bar{t}_{\uparrow} = [\downarrow \downarrow, \uparrow \! {*} \textcolor{red}{\pmb{)}}$ &
\\
\hdashline
$\tfrac{-1}{3}$ & $d_{\uparrow} = \textcolor{red}{\pmb{(}} 00, \uparrow \! 0 ]$ & $s_{\uparrow} = \textcolor{red}{\pmb{(}} \! \! \downarrow \uparrow, \uparrow \! 0 ]$ & $b_{\uparrow} = \textcolor{red}{\pmb{(}} \! \! \downarrow \downarrow, \uparrow \! 0 ]$ & \\
&  $d_{\downarrow} = \textcolor{red}{\pmb{(}} 00, \downarrow \! 0 ]$ & $s_{\downarrow} = \textcolor{red}{\pmb{(}} \! \! \uparrow \downarrow, \downarrow \! 0 ]$ & $b_{\downarrow} = \textcolor{red}{\pmb{(}} \! \! \uparrow \uparrow, \downarrow \! 0 ]$ & \\
\hdashline
$\tfrac{1}{3}$ & $\bar{d}_{\downarrow} = [00, 0 \! \uparrow \! \! \textcolor{red}{\pmb{)}}$ & $\bar{s}_{\downarrow} = [\downarrow \uparrow, 0 \! \uparrow \! \! \textcolor{red}{\pmb{)}}$ & $\bar{b}_{\downarrow} = [\downarrow \downarrow, 0 \! \uparrow \! \! \textcolor{red}{\pmb{)}}$ &\\
& $\bar{d}_{\uparrow} = [00, 0 \! \downarrow \! \! \textcolor{red}{\pmb{)}}$ & $\bar{s}_{\uparrow} = [\uparrow \downarrow, 0 \! \downarrow \! \! \textcolor{red}{\pmb{)}}$ & $\bar{b}_{\uparrow} = [\uparrow \uparrow, 0 \! \downarrow \! \! \textcolor{red}{\pmb{)}}$ &\\
    \end{tabular}
  \end{center}
\end{table}

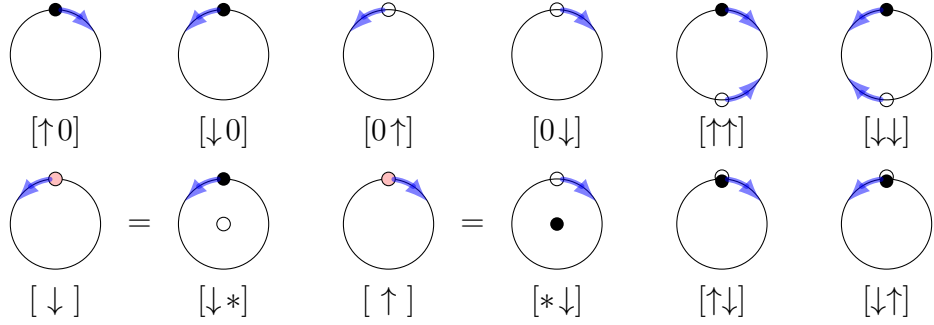
\begin{figure} \label{spin imped}
\begin{equation*}
\begin{array}{ccccccccccc}
\begin{tikzpicture}[scale=3,cap=round,>=latex]
  \draw (0cm,0cm) circle[radius=.2cm];
\fill[radius=.03cm] (0cm,.2cm) circle[];
   \draw[blue, line width=.7mm, opacity=.5,->] (0cm,0cm)+(83:.2cm) arc[start angle=83, end angle=30, radius=.2cm];
 \draw (0cm,-.225cm) node[anchor=north] {${[ \uparrow \! 0]}$};
\end{tikzpicture}
&&
\begin{tikzpicture}[scale=3,cap=round,>=latex]
  \draw (0cm,0cm) circle[radius=.2cm];
\fill[radius=.03cm] (0cm,.2cm) circle[];
   \draw[blue, line width=.7mm, opacity=.5,->] (0cm,0cm)+(97:.2cm) arc[start angle=97, end angle=150, radius=.2cm];
 \draw (0cm,-.225cm) node[anchor=north] {${[ \downarrow \! 0]}$};
\end{tikzpicture}
&&
\begin{tikzpicture}[scale=3,cap=round,>=latex]
  \draw (0cm,0cm) circle[radius=.2cm];
  \draw (0cm,.2cm) circle[radius=.03cm];
   \draw[blue, line width=.7mm, opacity=.5,->] (0cm,0cm)+(97:.2cm) arc[start angle=97, end angle=150, radius=.2cm];
 \draw (0cm,-.225cm) node[anchor=north] {${[0 \! \uparrow]}$};
\end{tikzpicture}
&&
\begin{tikzpicture}[scale=3,cap=round,>=latex]
  \draw (0cm,0cm) circle[radius=.2cm];
 \draw (0cm,.2cm) circle[radius=.03cm];
   \draw[blue, line width=.7mm, opacity=.5,->] (0cm,0cm)+(83:.2cm) arc[start angle=83, end angle=30, radius=.2cm];
 \draw (0cm,-.225cm) node[anchor=north] {${[0 \! \downarrow]}$};
\end{tikzpicture}
& &
\begin{tikzpicture}[scale=3,cap=round,>=latex]
  \draw (0cm,0cm) circle[radius=.2cm];
  \draw (0cm,-.2cm) circle[radius=.03cm];
\fill[radius=.03cm] (0cm,.2cm) circle[];
   \draw[blue, line width=.7mm, opacity=.5,->] (0cm,0cm)+(83:.2cm) arc[start angle=83, end angle=30, radius=.2cm];
   \draw[blue, line width=.7mm, opacity=.5,->] (0cm,0cm)+(-83:.2cm) arc[start angle=-83, end angle=-30, radius=.2cm];
 \draw (0cm,-.225cm) node[anchor=north] {$[ \uparrow \uparrow]$};
\end{tikzpicture}
& 
\raisebox{3.1 em}{ \ }
&
\begin{tikzpicture}[scale=3,cap=round,>=latex]
  \draw (0cm,0cm) circle[radius=.2cm];
  \draw (0cm,-.2cm) circle[radius=.03cm];
\fill[radius=.03cm] (0cm,.2cm) circle[];
   \draw[blue, line width=.7mm, opacity=.5,->] (0cm,0cm)+(97:.2cm) arc[start angle=97, end angle=150, radius=.2cm];
   \draw[blue, line width=.7mm, opacity=.5,->] (0cm,0cm)+(-97:.2cm) arc[start angle=-97, end angle=-150, radius=.2cm];
 \draw (0cm,-.225cm) node[anchor=north] {$[ \downarrow \downarrow]$};
\end{tikzpicture}
\\
\begin{tikzpicture}[scale=3,cap=round,>=latex]
  \draw (0cm,0cm) circle[radius=.2cm];
\fill[pink, radius=.03cm] (0cm,.2cm) circle[];
\draw (0cm,.2cm) circle[radius=.03cm];
   \draw[blue, line width=.7mm, opacity=.5,->] (0cm,0cm)+(97:.2cm) arc[start angle=97, end angle=150, radius=.2cm];
 \draw (0cm,-.225cm) node[anchor=north] {$[ \ \downarrow \ ]$};
\end{tikzpicture}
& 
\raisebox{3.1 em}{$=$}
&
\begin{tikzpicture}[scale=3,cap=round,>=latex]
  \draw (0cm,0cm) circle[radius=.2cm];
  \draw (0cm, 0cm) circle[radius=.03cm];
\fill[radius=.03cm] (0cm,.2cm) circle[];
   \draw[blue, line width=.7mm, opacity=.5,->] (0cm,0cm)+(97:.2cm) arc[start angle=97, end angle=150, radius=.2cm];
 \draw (0cm,-.225cm) node[anchor=north] {$[ \downarrow \! {*}]$};
\end{tikzpicture}
&
\raisebox{3.1 em}{ \ }
&
\begin{tikzpicture}[scale=3,cap=round,>=latex]
  \draw (0cm,0cm) circle[radius=.2cm];
\fill[pink, radius=.03cm] (0cm,.2cm) circle[];
\draw (0cm,.2cm) circle[radius=.03cm];
  \draw[blue, line width=.7mm, opacity=.5,->] (0cm,0cm)+(83:.2cm) arc[start angle=83, end angle=30, radius=.2cm];
 \draw (0cm,-.225cm) node[anchor=north] {$[ \ \uparrow \  ]$};
\end{tikzpicture}
& 
\raisebox{3.1 em}{$=$}
&
\begin{tikzpicture}[scale=3,cap=round,>=latex]
  \draw (0cm,0cm) circle[radius=.2cm];
\fill[radius=.03cm] (0cm,0cm) circle[];
 \draw (0cm,.2cm) circle[radius=.03cm];
   \draw[blue, line width=.7mm, opacity=.5,->] (0cm,0cm)+(83:.2cm) arc[start angle=83, end angle=30, radius=.2cm];
 \draw (0cm,-.225cm) node[anchor=north] {$[{*} \! \downarrow]$};
\end{tikzpicture}
&
\raisebox{3.1 em}{ \ }
&
\begin{tikzpicture}[scale=3,cap=round,>=latex]
  \draw (0cm,0cm) circle[radius=.2cm];
\draw (0cm,.21cm) circle[radius=.03cm];
\fill[radius=.03cm] (0cm,.19cm) circle[];
 \draw (0cm,-.225cm) node[anchor=north] {$[ \uparrow \downarrow]$};
   \draw[blue, line width=.7mm, opacity=.5,->] (0cm,0cm)+(83:.2cm) arc[start angle=83, end angle=30, radius=.2cm];
\end{tikzpicture}
 & &
\begin{tikzpicture}[scale=3,cap=round,>=latex]
  \draw (0cm,0cm) circle[radius=.2cm];
\draw (0cm,.21cm) circle[radius=.03cm];
\fill[radius=.03cm] (0cm,.19cm) circle[];
 \draw (0cm,-.225cm) node[anchor=north] {$[ \downarrow \uparrow]$};
   \draw[blue, line width=.7mm, opacity=.5,->] (0cm,0cm)+(97:.2cm) arc[start angle=97, end angle=150, radius=.2cm];
\end{tikzpicture}
\end{array}
\end{equation*}
\caption{All geom orbitals, depicted as spinor particles.
The spin-$0$ field $*$ in ${[\downarrow \! {*}]}$ and ${[{*} \! \downarrow]}$ is indicated pictorially by a spinor particle with $r = 0$.}
\end{figure}

The article is organized as follows.
In Section \ref{spinor chirality section} we show that on an internal spacetime, the two direct summands of the standard decomposition 
\begin{equation*}
\mathfrak{so}(1,3)_{\mathbb{C}} \cong \mathfrak{su}(2)_{\mathbb{C}} \oplus \mathfrak{su}(2)_{\mathbb{C}}
\end{equation*}
correspond to positive and negative electric charges. 
Consequently, a $\gamma^5$ eigenspinor with eigenvalue $1$ (resp.\ $-1$) has negative (resp.\ positive) electric charge.
We will find that this new identification leads to significant and novel changes in the meaning of the Dirac equation.

In Section \ref{4-momentum section} a `spinor particle' of a pointon is introduced to model the pointon's $4$-momentum. 
A spinor is then obtained as a column of the gamma matrix of the spinor particle's $4$-velocity under the internal metric $h$, with the column determined by the pointon's charge. 
Moreover, the speed (or equivalently, radius) of the spinor particle is a new parameter that does not appear using the normalized spinor alone. 
In Section \ref{spin vector section} we show that this new parameter allows for a classical description of off shell particles for which relativity is never violated.

Finally, in Section \ref{Lagrangian section} a Dirac-like Lagrangian is constructed whose equations of motion describe both the coupling of pointons of opposite charge, as well as pointon pair creation and annihilation.
From the combinatorics of this Lagrangian a composite model of the standard model particles is obtained, given in Table \ref{table1}.
To note, our model yields the correct spin, electric charge, and color charge of each standard model particle, and precisely three generations of leptons and quarks.
Our model also predicts the existence of a new massive spin-$2$ boson.
In the companion article \cite{B1}, we investigate standard model interaction vertices and particle masses in this framework.

\begin{Remark} \rm{
There are interesting proposals by Klinkhamer, Battista, and others, where degenerate spacetime metrics are used to regularize the big bang singularity, and it is suggested that these degeneracies have a quantum origin \cite{K1, K2, K3, KW, Ba}.
Degenerate spacetime metrics also arise in the context of loop quantum gravity; see \cite{LW} and references therein. 
}\end{Remark}

\textbf{Notation:} Tensors labeled with upper and lower indices $a,b, \ldots$ represent covector and vector slots respectively in Penrose's abstract index notation (so $v^a \in V$ and $v_a \in V^*$), and tensors labeled with indices $\mu, \nu, \ldots$ denote components with respect to a coordinate basis.
Given a curve $\beta: I \to \tilde{M}$, we often denote its image $\beta(I)$ also by $\beta$.
Natural units $\hbar = c = G = 1$ and the signature $(+,-,-,-)$ are used throughout.

We will use both the chiral and Dirac bases for spinors, and so to specify the basis the subscripts $\mathscr{C}$ and $\mathscr{D}$ will be used.
The change-of-basis matrices are 
\begin{equation} \label{chiral - Dirac}
\tfrac{1}{\sqrt{2}} \! \left[ \begin{smallmatrix} \bm{1} & -\bm{1} \\ \bm{1} & \bm{1} \end{smallmatrix} \right] \psi_{\mathscr{D}} = \psi_{\mathscr{C}} 
\ \ \ \ \text{ and } \ \ \ \ 
\tfrac{1}{\sqrt{2}} \! \left[ \begin{smallmatrix} \bm{1} & \bm{1} \\ -\bm{1} & \bm{1} \end{smallmatrix} \right] \psi_{\mathscr{C}} = \psi_{\mathscr{D}}.
\end{equation}
The gamma matrices in the two bases are 
\begin{equation*}
\gamma^0 = \left[\begin{smallmatrix} \bm{1} &  \\ & -\bm{1} \end{smallmatrix} \right]_{\mathscr{D}} 
= \left[ \begin{smallmatrix} & \bm{1} \\ \bm{1} & \end{smallmatrix} \right]_{\mathscr{C}}, 
\ \ \ \ 
\gamma^j = \left[ \begin{smallmatrix} & \sigma^j \\ - \sigma^j & \end{smallmatrix} \right]_{\mathscr{D}} 
= \left[ \begin{smallmatrix} & \sigma^j \\ -\sigma^j & \end{smallmatrix} \right]_{\mathscr{C}},
\ \ \ \
\gamma^5 = \left[ \begin{smallmatrix} & \bm{1} \\ \bm{1} & \end{smallmatrix} \right]_{\mathscr{D}} 
= \left[ \begin{smallmatrix} -\bm{1} & \\ & \bm{1} \end{smallmatrix} \right]_{\mathscr{C}},
\end{equation*}
where $\bm{1}$ denotes the $2 \times 2$ identity matrix, and
\begin{equation*}
\sigma^1 := \left[ \begin{smallmatrix} & 1 \\ 1 & \end{smallmatrix} \right], \ \ \ \
\sigma^2 := \left[ \begin{smallmatrix} & -i \\ i & \end{smallmatrix} \right], \ \ \ \
\sigma^3 := \left[ \begin{smallmatrix} 1 & \\ & -1 \end{smallmatrix} \right], \ \ \ \
(\gamma^{\mu})_{\mathscr{C}} = \tfrac{1}{\sqrt{2}} \! \left[ \begin{smallmatrix} \bm{1} & -\bm{1} \\ \bm{1} & \bm{1} \end{smallmatrix} \right] (\gamma^{\mu})_{\mathscr{D}} \tfrac{1}{\sqrt{2}} \! \left[ \begin{smallmatrix} \bm{1} & \bm{1} \\ -\bm{1} & \bm{1} \end{smallmatrix} \right].
\end{equation*}

\section{Spinor chirality on an internal spacetime} \label{spinor chirality section}

\subsection{Review of spinors}

Denote by $L := \operatorname{O}(1,3)$ the Lorentz group, and by $L_0 := \operatorname{SO}^+(1,3) \subset L$ the subgroup of proper orthochronous Lorentz transformations.

Following \cite{F}, consider the two $\mathbb{R}$-linear maps 
\begin{equation} \label{v maps}
\begin{array}{rcll}
\mathbb{R}^{1,3} & \longrightarrow & M_2(\mathbb{C}) & \\
v & \mapsto & v_* := v_0 \bm{1} + v_j \sigma^j & = \left[ \begin{smallmatrix} v_0 + v_3 & v_1 - i v_2 \\ v_1 + i v_2 & v_0 - v_3 \end{smallmatrix} \right] \\
v & \mapsto & v^* := v_0 \bm{1} - v_j \sigma^j & = \left[ \begin{smallmatrix} v_0 - v_3 & -v_1 + i v_2 \\ -v_1 - i v_2 & v_0 + v_3 \end{smallmatrix} \right] 
\end{array}
\end{equation}
These maps define the gamma matrices in the chiral basis $\mathscr{C}$,
\begin{equation} \label{gamma matrix map}
\begin{array}{rcl}
\gamma: \mathbb{R}^{1,3} 
& \longrightarrow & M_4(\mathbb{C})\\
v & \mapsto & \left[ \begin{smallmatrix} & v_* \\ v^* & \end{smallmatrix} \right]_{\mathscr{C}} = \gamma(v^{\mu}e_{\mu})  = v^{\mu} \gamma(e_{\mu}) = v^{\mu} \gamma_{\mu} = \slashed v
\end{array}
\end{equation} 
where $\gamma^{\mu}  := \gamma(e^{\mu})$.
The maps (\ref{v maps}) also determine a $2 \! : \! 1$ group homomorphism 
\begin{equation} \label{sl to l} 
\begin{array}{rcl}
\operatorname{Sl}(2, \mathbb{C}) & \stackrel{2:1}{\longrightarrow} & L_0\\
\pm s & \mapsto & \Lambda_{s}
\end{array}
\end{equation}
defined by\footnote{The map (\ref{sl to l}) is well defined: First observe that $\operatorname{det}v_* =\operatorname{det}v^* = v^av_a$ (and $v^av_a \bm{1} = v^* v_* = v_* v^*$).  Whence, $(\Lambda_s v)^a (\Lambda_sv)_a = \operatorname{det} (\Lambda_s v)_* = \operatorname{det} (s v_* s^{\dagger}) = \operatorname{det} s \operatorname{det} v_* \operatorname{det} s^{\dagger} = \operatorname{det} v_* = v^a v_a$.}$^,$\footnote{
Note that (\ref{bb}) may be written in the (physics textbook) form $\rho(s) \tensor{\Lambda}{^{\mu}_{\nu}} \gamma_{\mu} = \gamma_{\nu} \rho(s)$ since it implies $\gamma(\Lambda_s v) = \rho(s) \gamma(v) \rho(s)^{-1}$, and $\gamma(\Lambda v) = \gamma(\tensor{\Lambda}{^{\mu}_{\nu}}v^{\nu}) = \tensor{\Lambda}{^{\mu}_{\nu}} v^{\nu} \gamma(e_{\mu}) = \tensor{\Lambda}{^{\mu}_{\nu}} v^{\nu} \gamma_{\mu}$.} 
\begin{equation} \label{bb}
(\Lambda_s v)_* = s v_* s^{\dagger}, \ \ \ \text{ or equivalently, } \ \ \ (\Lambda_s v)^* = s^{\dagger -1} v^* s^{-1}.
\end{equation}
Furthermore, $\operatorname{Sl}(2,\mathbb{C})$ admits the representation 
\begin{equation} \label{rho rep}
\begin{array}{rcl}
\rho: \operatorname{Sl}(2, \mathbb{C}) & \longrightarrow & \operatorname{Gl}(4, \mathbb{C})\\
s & \mapsto & \left[ \begin{smallmatrix} s & \\ & s^{\dagger -1} \end{smallmatrix} \right]_{\mathscr{C}}
\end{array} 
\end{equation}
In this article a \textit{spinor field} on a locally flat open set $U \subset \tilde{M}$ is a map $\psi: U \to \mathbb{C}^4$, smooth on its support, such that a Lorentz transformation $\Lambda \in L_0$ acts by 
\begin{equation*}
\Lambda \psi(x) := \rho(s) \psi(\Lambda x),
\end{equation*}
where $s \in \operatorname{Sl}(2, \mathbb{C})$ is one of the two elements for which $\Lambda = \Lambda_s$.

\begin{Remark} \label{sign of phase} \rm{
It is standard to identify the sign $\mp$ of the phase $e^{\mp i p_ax^a}$ of a plane-wave spinor $\psi = \left[ \begin{smallmatrix} * \\ * \\ * \\ * \end{smallmatrix} \right] e^{\mp i p_ax^a} \in \operatorname{ker}( i \slashed \partial - m)$ with its description as a particle or an antiparticle, since the sign of the phase appears to coincide with the sign of the energy $\pm |E| = \pm m$ (or time orientation):
\begin{equation*}
\begin{split} 
0 & = (i \slashed \partial - m) u(p)e^{-ip_ax^a} = (\slashed p - m) u(p)e^{-ip_ax^a},\\
0 & = (i \slashed \partial - m) v(p)e^{ip_ax^a} = (- \slashed p - m) v(p)e^{ip_ax^a}.
\end{split}
\end{equation*}
However, there is a space of positive energy solutions spanned by
\begin{equation*} \label{positive}
u(p)e^{-ip_ax^a} = \left[ \begin{matrix} \eta \\ \tfrac{\vec{\sigma} \cdot \vec{p}}{E+m} \eta \end{matrix} \right]_{\mathscr{D}} \! \! e^{-ip_ax^a}, \ \ \ \ \
v(p)e^{ip_ax^a} = \left[ \begin{matrix} \tfrac{\vec{\sigma} \cdot \vec{p}}{E+m} \eta \\ \eta \end{matrix} \right]_{\mathscr{D}} \! \! e^{ip_ax^a},
\end{equation*}
with $\eta \in \mathbb{C}^2$, and a space of negative energy solutions spanned by
\begin{equation*} \label{negative}
\tilde{u}(p)e^{-ip_ax^a} = \left[ \begin{matrix} \tfrac{\vec{\sigma} \cdot \vec{p}}{E-m} \eta \\ \eta \end{matrix} \right]_{\mathscr{D}} \! \! e^{-ip_ax^a}, \ \ \ \ \
\tilde{v}(p)e^{ip_ax^a} = \left[ \begin{matrix} \eta \\ \tfrac{\vec{\sigma} \cdot \vec{p}}{E-m} \eta \end{matrix} \right]_{\mathscr{D}} \! \! e^{ip_ax^a}.
\end{equation*}
\textit{In each solution space, both phase signs occur equally.}\footnote{Indeed, in the rest frame of a particle with a positive phase we may have, for example,
\begin{equation*}
(i \slashed \partial - m) \left[ \begin{smallmatrix} 0 \\ 0 \\ 1 \\ 0 \end{smallmatrix} \right]_{\mathscr{D}} \! \! e^{iEt} = 
\text{\tiny{$\begin{bmatrix} 0 & 0 \\ 0 & E - m \end{bmatrix}$}} 
\left[ \begin{smallmatrix} 0 \\ 0 \\ 1 \\ 0 \end{smallmatrix} \right]_{\mathscr{D}} \! \! e^{iEt} = \left[ \begin{smallmatrix} 0 \\ 0 \\ E-m \\ 0 \end{smallmatrix} \right]_{\mathscr{D}} \! \! e^{iEt}.
\end{equation*}
This expression equals zero if and only if $E = m > 0$.} 
Consequently, the sign of the phase $e^{\pm ip_ax^a}$ of a Dirac spinor is mathematically unrelated to its representation as a particle or antiparticle.
Thus, the identification of phase sign with a particle/antiparticle state is at best interpretational, and it is not an interpretation we assume here.
}\end{Remark}

\subsection{Spinor chirality = electric charge} \label{chirality is charge section}

Recall that the choice of time orientation $o_0 \in \{ \pm 1\}$  along the worldline $\beta \subset \tilde{M}$ of a pointon is identified with the pointon's electric charge (though time does not flow along $\beta$ in spacetime $M$); see Section \ref{intro}.

\begin{Proposition} \label{keep}
A choice of time orientation $o_0 \in \{ \pm 1 \}$ of $\mathbb{R}^{1,3}$, and thus of electric charge of a pointon, corresponds to a unique $\mathfrak{su}(2)_{\mathbb{C}} := \mathfrak{su}(2) \oplus i \mathfrak{su}(2)$ subalgebra of the complexified Lorentz algebra,
\begin{equation} \label{decomp}
\mathfrak{so}(1,3)_{\mathbb{C}} \cong \mathfrak{su}(2)_{\mathbb{C}} \oplus \mathfrak{su}(2)_{\mathbb{C}}.
\end{equation}
\end{Proposition}

\begin{proof}
For $i \in \{ 1,2,3 \}$, denote by $J_i, K_i \in \mathfrak{so}(1,3)$ the Lie algebra generators for rotations and boosts respectively in the direction $e_i$.
Recall that the linear combinations
\begin{equation*}
A_i := \tfrac 12 (J_i - iK_i) \ \ \ \text{ and } \ \ \ B_i := \tfrac 12 (J_i + i K_i)
\end{equation*}
generate the left and right $\mathfrak{su}(2)_{\mathbb{C}}$ subalgebras of the direct sum decomposition (\ref{decomp}).

Observe that under time orientation reversal $o_0 \mapsto -o_0$,
\begin{itemize} 
 \item[(i)] rotations $\exp (s a^iJ_i) \in L_0$ are invariant since they have no time dependence\footnote{Rotation here means rotation of reference frame, and should not to be confused with something rotating \textit{in time}.};
 \item[(ii)] boosts $\exp (s a^iK_i) \in L_0$ are replaced by their inverses:
\begin{equation*}
\exp (s K_i) \mapsto \exp (s K_i)^{-1} = \exp (-s K_i).
\end{equation*}
\end{itemize}
For example, consider a boost in the $x$-direction
\begin{equation*}
\Lambda = \left[ \begin{smallmatrix} \gamma & v \gamma & & \\ v \gamma & \gamma & & \\ && 1 & \\ &&& 1 \end{smallmatrix} \right],
\end{equation*}
where $\gamma = (1- v^2)^{-1/2}$.
Then $o_0 \mapsto -o_0$ yields $v \mapsto -v$ and $\gamma \mapsto \gamma$, whence 
\begin{equation*}
\Lambda \mapsto \Lambda' = \left[ \begin{smallmatrix} \gamma & -v \gamma & & \\ -v \gamma & \gamma & & \\ && 1 & \\ &&& 1 \end{smallmatrix} \right].
\end{equation*}
But $\Lambda \Lambda' = \Lambda' \Lambda = \mathbf{1}_4$, and so $\Lambda' = \Lambda^{-1}$.

Reversing the time orientation thus transforms the generators by
\begin{equation*}
J_i \mapsto J_i \ \ \ \ \text{ and } \ \ \ \ K_i \mapsto -K_i.
\end{equation*}
Therefore time orientation reversal swaps $A_i$ and $B_i$: $A_i \mapsto B_i \mapsto A_i$.
\end{proof}

Consider a spinor $\psi \in \mathbb{C}^4$ that transforms under a chiral representation of $\mathfrak{so}(1,3)$, that is, under one of the irreducible subrepresentations of $\rho$ in (\ref{rho rep}),
\begin{equation*}
\begin{array}{rrcl}
& \operatorname{Sl}(2,\mathbb{C}) & \longrightarrow & \operatorname{Gl}(4,\mathbb{C}) \\
\rho^-: & s & \mapsto & \left[ \begin{smallmatrix} s & 0 \\ 0 & 0 \end{smallmatrix} \right]_{\mathscr{C}} \\
\rho^+: & s & \mapsto & \left[ \begin{smallmatrix} 0 & 0 \\ 0 & s^{\dagger -1} \end{smallmatrix} \right]_{\mathscr{C}}
\end{array} 
\end{equation*} 
(These are the $(\tfrac 12, 0)$ and $(0, \tfrac 12)$ representations of $\operatorname{Sl}(2,\mathbb{C})$.)
The Lie algebra isomorphisms $\mathfrak{so}(1,3) \cong \mathfrak{sl}(2,\mathbb{C}) \cong \mathfrak{su}(2)_{\mathbb{C}}$, together with the relation $A_i^{\dagger -1} = A_i^* = B_i$, then imply that the time orientation $o_0 = \pm 1$ corresponds to the representation $\rho^{\pm}$, by Proposition \ref{keep}. 
Consequently, \textit{if $\psi$ represents a pointon on an internal spacetime $M$, then the chirality of $\psi$ is its electric charge.}

\begin{Corollary} \label{charge corollary}
If $\psi \in \mathbb{C}^4$ is an eigenspinor of $\gamma^5$ with eigenvalue $\pm 1$, then $\psi$ has electric charge $\pm 1$; otherwise $\psi$ is neutral.
In the latter case, the projections
\begin{equation*} \label{two projections}
\psi^- := \tfrac 12(1 - \gamma^5) \psi = \left[ \begin{smallmatrix} * \\ * \\ 0 \\ 0 \end{smallmatrix} \right]_{\mathscr{C}} \ \ \ \ \text{ and } \ \ \ \ \psi^+ := \tfrac 12 (1 + \gamma^5) \psi = \left[ \begin{smallmatrix} 0 \\ 0 \\ * \\ * \end{smallmatrix} \right]_{\mathscr{C}}
\end{equation*}
are spinors with charges $-1$ and $+1$, respectively.
\end{Corollary}

\begin{Remark} \rm{
Time reversal in the framework of internal spacetime is the reversal of time orientation, $o_0 \mapsto -o_0$.
Thus, time reversal and charge conjugation cannot be separated into two distinct discrete symmetries: time reversal necessarily induces charge conjugation, and charge conjugation necessarily induces time reversal.
}\end{Remark}

\section{The spinor and internal $4$-momentum of a timelike pointon} \label{4-momentum section}

Fix an orientation of $\tilde{M}$.
Consider a pointon with worldline $\beta \subset \tilde{M}$, timelike $4$-velocity $v \in \tilde{M}_{\beta(t)}$, and (ontic) spin vector $s \in M_{\beta(t)}$.
Let $e_0, \ldots, e_3$ be a positive orthonormal tetrad parallel transported along $\beta$ for which $v = e_0$ and $s = o_0o_{12} e_3$, with $o_0, o_{12} \in \{ \pm 1 \}$ the free choices of orientation of $\operatorname{span}\{e_0\}$ and $\operatorname{span}\{e_1, e_2 \}$ respectively.

Since the pointon is a dust particle, it has a rest energy $\omega > 0$, or equivalently, $4$-momentum $k = \omega v$. 
We would like to identify $\omega$ with something geometric. 
Let $\beta(t) \in \tilde{M}$ be a point along $\beta$ that is not transversely intersected by other pointon worldlines, that is, a point where the internal tangent space $M_{\beta(t)}$ is $3$-dimensional.
Then
\begin{equation*}
M_{\beta(t)} = \operatorname{span}\{e_1, e_2, e_3\} \subset \tilde{M}_{\beta(t)}.
\end{equation*}
The internal $4$-velocity of the pointon at $\beta(t)$ is thus
\begin{equation*}
\breve{v}_{abc} = o_0 \star \! v_d = o_0 e^1 \wedge e^2 \wedge e^3.
\end{equation*}
Consequently, the \textit{internal $4$-momentum}  of the pointon is
\begin{equation*}
\breve{k}_{abc} = o_0 \star \! k_d = o_0 \star \! \omega v_d = \omega \breve{v}_{abc}.
\end{equation*}

Now consider the contraction of $\breve{k}_{abc}$ with the spin vector $s = o_0 o_{12} e_3$,
\begin{equation*} \label{sk contraction}
s^a \breve{k}_{abc} = o_0^2 o_{12} e^1 \wedge e^2 \wedge e^3 (e_3) = o_{12} \omega e^1 \wedge e^2.
\end{equation*}
Denote by $\hat{\star}$ the Hodge star operator in the internal tangent space $M_{\beta(t)}$.

\begin{Lemma} \label{internal Hodge}
The internal Hodge dual of the contraction $s^a \breve{k}_{abc}$ is the vector
\begin{equation} \label{contract}
\hat{\star} s^a \breve{k}_{abc} = \omega s,
\end{equation}
and therefore may be interpreted as an angular frequency vector in the rest frame of the pointon.
\end{Lemma}

\begin{proof}
It suffices to set the (nonphysical) orientation $o_{\tilde{M}} = o_{0123}$ of the external tangent space $\tilde{M}_{\beta(t)}$ to $1$.
Then
\begin{equation*}
\hat{\star} s^a \breve{k}_{abc} = o_{12} \omega (o_{123}e_1 \times e_2) = o_{12}(o_{0123}o_{0}^{-1}) \omega e_3 = \omega s.
\end{equation*}
Indeed, the orientation $o_{123}$ of $M_{\beta(t)} = \operatorname{span}\{ e_1, e_2, e_3 \}$ arises from the Hodge star operator $\hat{\star}$ on $M_{\beta(t)}$.
(We may suppose that $o_{123} = 1$ and $o_{123} = -1$ correspond to the right-hand rule and left-hand rule respectively.)
Furthermore, with our choice $o_{0123} = 1$, we have $o_{123} = o_{0123}o_0^{-1} = o_{0123}o_0 = o_0$.
\end{proof}

It follows from (\ref{contract}) that the two free orientations $o_0$ and $o_{12}$ in the spin vector $s = o_0o_{12}e_3$ determine the direction of rotation of the angular frequency vector 
$\hat{\star} s^a \breve{k}_{abc}$, since $\omega$ is positive. 

The question, then, is what is rotating with angular frequency vector $\omega s$?
A simplest proposal is a particle with a circular trajectory about $\beta(t) = (t, 0, 0, 0)$,
\begin{equation} \label{spinon}
\alpha(t) = \alpha^{\mu} = (t, r \sin (\omega t), o_{12} r \cos (\omega t), 0)
\end{equation}
for some choice of radius $r > 0$, in the normal coordinates induced from $e_0, \ldots, e_3$.

Recall the gamma map $w \mapsto \gamma(w) = \slashed w$ in (\ref{gamma matrix map}) and the internal metric $h$ in (\ref{im}).

\begin{Definition} \rm{
We call the particle $\alpha(t)$ in (\ref{spinon}) the \textit{spinor particle} of the pointon, and define the pointon's spinor $\psi(\beta(t)) \in \mathbb{C}^4$ to be one of the four columns of the gamma matrix of the internal metric projection $h$ of the $4$-velocity $\dot{\alpha}^{\mu} := \tfrac{d}{d\tau}\alpha^{\mu}$,
\begin{equation*}
\gamma(h(\dot{\alpha})) = \gamma^{\mu}h_{\mu \nu} \dot{\alpha}^{\nu}.
\end{equation*}
}\end{Definition}

\begin{Lemma} \label{ildk}
The spinors of a charged pointon with spin $o_{12} \in \{ +1, -1 \} = \{ \uparrow, \downarrow \}$ in the $e_3$ direction in its rest frame are
\begin{equation} \label{spinor reps}
u \! \left[ \begin{smallmatrix} 1 \\ 0 \\ 0 \\ 0 \end{smallmatrix} \right]_{\mathscr{C}} e^{o_{12} i \omega t}, \ \ \ 
u \! \left[ \begin{smallmatrix} 0 \\ 1 \\ 0 \\ 0 \end{smallmatrix} \right]_{\mathscr{C}} e^{-o_{12} i \omega t}, \ \ \ 
u \! \left[ \begin{smallmatrix} 0 \\ 0 \\ -1 \\ 0 \end{smallmatrix} \right]_{\mathscr{C}} e^{-o_{12} i \omega t}, \ \ \ 
u \! \left[ \begin{smallmatrix} 0 \\ 0 \\ 0 \\ -1 \end{smallmatrix} \right]_{\mathscr{C}} e^{o_{12} i \omega t}.
\end{equation}
\end{Lemma}

\begin{proof}
We have 
\begin{equation*}
\dot{\alpha} = e_0 + u \cos (\omega t) e_1 - o_{12} u \sin(\omega t) e_2,
\end{equation*} 
by (\ref{spinon}), noting that $t$ is the proper time parameter $\tau$ of $\beta$.
Thus 
\begin{equation*}
h(\dot{\alpha}) =  \tensor{h}{^{\mu}_{\nu}}\dot{\alpha}^{\nu} = u \cos (\omega t)e_1 - o_{12} u \sin (\omega t)e_2,
\end{equation*} 
by (\ref{im}).
Consequently, by (\ref{v maps}), 
\begin{equation*}
\begin{split}
h(\dot{\alpha})_* & = u \! \left[ \begin{matrix} 0 & \cos (\omega t) + o_{12} i \sin (\omega t) \\ \cos (\omega t) - o_{12} i \sin (\omega t) & 0 \end{matrix}\right] \\
& = u \! \left[ \begin{matrix} 0 & e^{o_{12} i \omega t} \\ e^{-o_{12} i \omega t} & 0 \end{matrix} \right].
\end{split}
\end{equation*}
Furthermore, observe that $h(\dot{\alpha})^* = -h(\dot{\alpha})_*$.
Therefore, by (\ref{gamma matrix map}),
\begin{equation*}
\gamma(h(\dot{\alpha})) = \left[ \begin{matrix} 0 & h(\dot{\alpha})_* \\ h(\dot{\alpha})^* & 0 \end{matrix} \right]_{\mathscr{C}} 
= u \! \left[ \begin{smallmatrix} &&& e^{o_{12} i \omega t} \\ && e^{-o_{12} i \omega t} & \\ & -e^{o_{12} i \omega t} && \\ -e^{-o_{12} i \omega t} &&& \end{smallmatrix} \right]_{\mathscr{C}}.
\end{equation*}
\end{proof}

We are free to rescale the spinors (\ref{spinor reps}) by a nonzero complex number \textit{as long as all the spinors (\ref{spinor reps}) are rescaled by the same number}.
Consequently, we do not remove the minus signs from the latter two. 
In particular, we do not normalize each spinor separately. 

\begin{Definition} \rm{
We say a set of spinors $\mathcal{O}$ is \textit{covariantly independent} if for any two spinors $\psi_1, \psi_2 \in \mathcal{O}$, we have $\psi_1 = \psi_2$ whenever there is a $c \in \mathbb{C}$ and $\Lambda \in L_0$ such that $\psi_1 = c \Lambda \psi_2$.
}\end{Definition}

\begin{Proposition} \label{ngu}
The spinors
\begin{equation*} \label{max set}
\left[ \begin{smallmatrix} e^{i \omega t} \\ 0 \\ 0 \\ 0 \end{smallmatrix} \right]_{\mathscr{C}} =: [ \uparrow \! 0], \ \ \ 
\left[ \begin{smallmatrix} e^{-i \omega t} \\ 0 \\ 0 \\ 0 \end{smallmatrix} \right]_{\mathscr{C}} =: [ \downarrow \! 0], \ \ \ 
\left[ \begin{smallmatrix} 0 \\ 0 \\ -e^{-i \omega t} \\ 0 \end{smallmatrix} \right]_{\mathscr{C}} =: [ 0 \! \uparrow ], \ \ \ 
\left[ \begin{smallmatrix} 0 \\ 0 \\ -e^{i \omega t} \\ 0 \end{smallmatrix} \right]_{\mathscr{C}} =: [ 0 \! \downarrow ]
\end{equation*}
form a maximal covariantly independent set of pointon $\gamma^5$ eigenspinors.
\end{Proposition}

\begin{proof}
Recall the maps (\ref{sl to l}) and (\ref{rho rep}).
The Lorentz transformation $\Lambda = \left[ \begin{smallmatrix} 1 & & & \\ & 1 & & \\ & & -1 & \\ & & & -1 \end{smallmatrix} \right]$ acts on spinors in both the Dirac and chiral representations by\footnote{Note that the factor $i = e^{i \pi/2}$ in (\ref{dfak}) simply rotates the phase of the spinor by $\pi/2$.}
\begin{equation} \label{dfak}
\rho(i \sigma^1) = i \left[ \begin{smallmatrix} \sigma^1 & \\ & \sigma^1 \end{smallmatrix} \right],
\end{equation}
since $\Lambda = \Lambda_{\pm i \sigma^1}$.
Equivalently, a Lorentz transformation $\exp \left(- \tfrac i2 \omega_{\mu \nu} J^{\mu \nu} \right) \in L_0$, where $\omega_{\mu \nu}$ is antisymmetric and $J^{\mu \nu}$ are the six generators of the Lorentz Lie algebra, acts on spinors by $\exp \left( - \tfrac i2 \omega_{\mu \nu} S^{\mu \nu} \right)$, where $S^{\mu \nu} := \tfrac i4 [ \gamma^{\mu},\gamma^{\nu} ]$.
Whence, $\exp \left(- i \omega J^{23} \right) \in L_0$ acts on spinors by 
\begin{equation*}
\exp \left( - i \omega S^{23} \right) = \exp \left(- \tfrac i2 \omega \left[ \begin{smallmatrix} \sigma^1 & \\ & \sigma^1 \end{smallmatrix} \right] \right)
= \left[ \begin{smallmatrix} \cos (\omega/2) & -i \sin (\omega/2) & & \\ -i \sin (\omega/2) & \cos (\omega/2) && \\ && \cos (\omega/2) & -i \sin (\omega/2) \\ & & -i \sin (\omega/2) & \cos (\omega/2) \end{smallmatrix} \right].
\end{equation*}
The matrix (\ref{dfak}) is then obtained by letting $\omega = -\pi$.

Consequently, the four pairs of spinors
\begin{equation*}
\left[ \begin{smallmatrix} e^{\pm i \omega t} \\ 0 \\ 0 \\ 0 \end{smallmatrix} \right]_{\mathscr{C}}, \ \ \ \left[ \begin{smallmatrix} 0 \\ e^{\pm i \omega t} \\ 0 \\ 0 \end{smallmatrix} \right]_{\mathscr{C}}; \ \ \ \ \text{ and } \ \ \ \ 
\left[ \begin{smallmatrix} 0 \\ 0 \\ -e^{\pm i \omega t} \\ 0 \end{smallmatrix} \right]_{\mathscr{C}}, \ \ \ \left[ \begin{smallmatrix} 0 \\ 0 \\ 0 \\ -e^{\pm i \omega t} \end{smallmatrix} \right]_{\mathscr{C}}
\end{equation*}
are each related by $\Lambda$.
\end{proof}

A maximal covariantly independent set of charged pointon spinors is therefore
\begin{equation*} \label{O0}
\mathcal{O}_{5} := \left\{ [a 0], [0a] \, | \, a \in \{ \uparrow, \downarrow\} \right\}.
\end{equation*}
A spinor in $\mathcal{O}_5$ is characterized by precisely two properties: 
\begin{itemize}
 \item the spin $\uparrow$ or $\downarrow$ of its pointon, specified by the sign $\pm$ of its phase $e^{\pm i \omega t}$; and
 \item the electric charge $-1$ or $+1$ of its pointon, specified respectively by 
\begin{equation*}
\left[ \begin{smallmatrix} 1 \\ 0 \\ 0 \\ 0 \end{smallmatrix} \right]_{\mathscr{C}} \ \ \ \text{ or } \ \ \ \left[ \begin{smallmatrix} 0 \\ 0 \\ -1 \\ 0 \end{smallmatrix} \right]_{\mathscr{C}}.
\end{equation*}
\end{itemize}

We conclude with a remark about left-hand and right-hand rules.
Recall from Lemma \ref{internal Hodge} that the spin vector $s$ in the $e_3$ direction is given by $s = o_0o_{12}e_3$.
Thus, the direction of rotation is determined by the product of the orientations $o_0$ and $o_{12}$. 
The timelike orientation $o_0$ (or equivalently, electric charge) induces an orientation $o_{123} = o_{0123}o_0^{-1}$ of the spatial subspace $\operatorname{span}\{e_1, e_2, e_3 \} \subset \tilde{M}_{\beta(t)}$, where the orientation $o_{0123}$ of $\tilde{M}_{\beta(t)}$ is arbitrary and nonphysical.
Therefore, if $o_0 = 1$ (resp.\ $o_0 = -1$) and $o_{0123} = 1$, then the right-hand (resp.\ left-hand) rule determines the direction of rotation of the spinor particle: the thumb points in the direction of the spin vector $s$, and the fingers curl in the direction of rotation.
Consequently, \textit{spinor particles of pointons of opposite charge and equal (resp.\ opposite) spin rotate in opposite (resp.\ equal) directions.}
This is shown in Figure \ref{spin imped}.

\section{Spinors and the mass shell condition} \label{spin vector section}

Consider a pointon with timelike worldline $\beta \subset \tilde{M}$ and spinor particle $\alpha(t)$ given in (\ref{spinon}).
In the following, we consider the physical role of the fixed radius $r > 0$ of the spinor particle $\alpha(t)$.

Denote by $u = \omega r$ the tangential speed of the spinor particle in the rest frame of the pointon.
The existence of $r$ determines a special frequency $\omega = \omega_*$, namely, the frequency for which $u$ is lightlike, $u = 1$:\footnote{Lightlike tangential speed may be viewed as a geodesic-like property: suppose a circle of radius $r$ is rotating with speed $u$ measured in an inertial frame.
In the accelerated frame of the circle, Ehrenfest observed that the circumference is
\begin{equation*}
C = 2\pi r \left( 1 - u^2 \right)^{-1/2} = 2 \pi r \gamma(u).
\end{equation*}
Thus, if $u = 1$, then $C$ is infinite. 
If a circle of infinite circumference is regarded as a straight line, then the particle travels in a `straight line' in its own reference frame if and only if it travels at light speed $u = 1$.}
\begin{equation*}
\omega_* = r^{-1}.
\end{equation*}

\begin{Definition} \label{mass def} \rm{
We identify the special frequency $\omega_*$ with the \textit{mass} $m$ of the pointon,
\begin{equation*} \label{m = r inverse}
m := \omega_* = r^{-1} = \hbar/(cr),
\end{equation*}
where units have been restored on the right.
}\end{Definition}

\begin{Remark} \rm{
In contrast to the standard definition of mass, we have \textit{not} defined the mass $m$ of a pointon to be the length of its $4$-momentum $k = \omega v$, and thus its rest energy $E_0 = \omega$.
The reason for this departure is so that we can describe off-shell particles \textit{classically}; indeed, the aim of internal spacetime is to be an approximate model of quantum theory using degenerate metrics. 
Recall that a particle is said to be off shell if $k^2 = k^ak_a \not = m^2$.
Such particles are purely quantum effects and occur at every trivalent electron-photon vertex in Feynman diagrams.\footnote{To see this, consider an electron-photon vertex with electron $4$-momenta $k_1, k_2$, and photon $4$-momentum $k_3 \not = 0$.
Assume to the contrary that all three particles are on shell: $k_1^2 = k_2^2 = m_e^2$ and $k_3^2 = 0$. By $4$-momentum conservation it suffices to suppose $k_1 + k_2 = k_3$, whence
\begin{equation*}
0 = k_3^2 = (k_1 + k_2)^2 = k_1^2 + k_2^2 + 2k_1^a k_{2a} = 2m_e^2 + 2k_1^a k_{2a}.
\end{equation*}
This implies $k_1^a k_{2a} = -m_e^2$, so $k_1 = -k_2$. 
But then $k_3 = 0$, contrary to assumption.} 
} \end{Remark}

The relations $m = r^{-1}$, $E_0 = \omega$, and $u = r \omega$, together imply
\begin{equation*} \label{E = mcu}
E_0 = \omega = ur^{-1} = mu = mcu.
\end{equation*}
This reduces to Einstein's mass-energy relation $E_0 = mc^2 = m$ whenever $u = c = 1$.
Thus, the $4$-momentum $k$ of any timelike pointon satisfies
\begin{equation*} \label{k squared}
k^2 = k^a k_a = E_0^2 = \omega^2 = m^2u^2,
\end{equation*}
whereas the relation
\begin{equation*}
k^2 = m^2
\end{equation*}
only holds if the speed $u$ of a pointon's spinor particle is lightlike, $u = 1$.
\textit{We therefore obtain a description of off-shell particles for which relativity is never violated:} 
\begin{itemize}
 \item A pointon is on shell, $k^2 = m^2$, if and only if $u = 1$.
 \item Regardless of whether a pointon is on or off shell, the relativistic constraint $k^2 = E_0^2$ always holds.
\end{itemize}

We review three derivations of $E_0 = mc^2$ in the Appendix, and show that no inconsistencies arise between these derivations and our new relation $E_0 = mcu$.
Indeed, each derivation assumes the particles are on shell, whence $u = c$. 

Finally, the relation $E_0 = mu$ (or equivalently, $k^a k_a - m^2u^2 = 0$) modifies the Dirac Lagrangian $\mathcal{L}_{\mathscr{D}} = \bar{\psi}(i \slashed \partial - m) \psi$ by replacing $m$ with $mu = r^{-1}u = \omega$:
\begin{equation*}
\mathcal{L} = \bar{\psi}(i \slashed \partial - mu) \psi = \bar{\psi}(i \slashed \partial - \omega) \psi.
\end{equation*}
By virtue of $u$ being a free parameter, off-shell spinors are thus solutions to the \textit{classical equations of motion} of $\mathcal{L}$, namely $(i \slashed \partial - mu) \psi = 0$.

\section{The Dirac Lagrangian on an internal spacetime} \label{Lagrangian section}

The pointon spinors $\psi \in \mathcal{O}_5 := \{ {[a0]}, {[0a]} \, | \, a \in \{ \uparrow, \downarrow \} \}$, given in Proposition \ref{ngu}, satisfy the Klein-Gordon equation, $(\partial^2 + \omega^2)\psi = (-i \slashed \partial - \omega)(i \slashed \partial - \omega) \psi = 0$.
Thus, the Klein-Gordon equation 
arises in our framework from the gamma representation of Minkowski spacetime (\ref{gamma matrix map}), together with the internal metric $h$.
In particular, we obtain the Klein-Gordon equation without assuming canonical quantization.

It was shown in \cite{B6} that pointons are spin-$\tfrac 12$ fermions.
However, it is readily verified that the pointon spinors $\psi \in \mathcal{O}_5$ do not satisfy the Dirac equation, $(i \slashed \partial - \omega) \psi \not = 0$. 
The fundamental meaning of the Dirac equation is therefore different in the context of internal spacetime.
In this section we introduce a Dirac-like Lagrangian for pointons and establish the meaning of the Dirac equation on internal spacetime.

Let $U \subset \tilde{M}$ be a locally flat open set.
Consider the subspace of the $\mathbb{C}$-vector space of pointon spinor fields $\psi: U \to \mathbb{C}^4$,
\begin{equation*} 
\left\langle \mathcal{O}_5 \right\rangle := \left\{ \sum_j c_j \Lambda_j \psi_j \, | \, c_j \in \mathbb{C}, \ \Lambda_j \in L_0, \ \psi_j \in \mathcal{O}_5 \right\}.
\end{equation*}
Let $\hat{\omega}: \left\langle \mathcal{O}_5 \right\rangle \to \left\langle \mathcal{O}_5 \right\rangle$ be the $\mathbb{C}$-linear map defined on pointon spinors $\psi \in \mathcal{O}_5$ with angular frequency $\omega = r^{-1}u = mu \geq 0$ and orientation $o_{12} \in \{ \pm 1 \}$ by 
\begin{equation*}
\hat{\omega} \psi = o_{12} \omega \psi,
\end{equation*}
and extend by $L_0$ to all pointon spinors: for $\psi \in \mathcal{O}_5$ and $\Lambda \in L_0$, set
\begin{equation} \label{nguxinfty}
\hat{\omega}(\Lambda \psi) := \Lambda \hat{\omega}(\psi).
\end{equation}
This operator is well-defined by Proposition \ref{ngu}.
Consider the Dirac-like operator
\begin{equation*}
\delta := i \slashed \partial - \hat{\omega}: \left\langle \mathcal{O}_5 \right\rangle \to \left\langle \mathcal{O}_5 \right\rangle.
\end{equation*}
We define the pointon Lagrangian density to be
\begin{equation*}
\mathcal{L} := \bar{\psi} \delta \psi = \bar{\psi}(i \slashed \partial - \hat{\omega}) \psi,
\end{equation*}
with $\bar{\psi} := \psi^{\dagger} \gamma^0$.
Similar to the usual Dirac Lagrangian, the two equations of motion of $\mathcal{L}$ are both $\delta \psi = 0$.
These equations are Lorentz invariant by (\ref{nguxinfty}):
\begin{equation*}
\delta (\Lambda \psi) = (\delta \Lambda) \psi = (\Lambda \delta ) \psi = \Lambda( \delta \psi ) = 0.
\end{equation*}

Recall from Corollary \ref{charge corollary} that a generic spinor $\psi: U \to \mathbb{C}^4$ decomposes into $\gamma^5$ eigenspinors $\psi = \psi^- + \psi^+$ of negative and positive electric charges, and consider the chiral decomposition of the Dirac Lagrangian, 
\begin{equation*} \label{chiral decomposition}
\mathcal{L} = \bar{\psi}(i \slashed \partial - \hat{\omega}) \psi = i \bar{\psi}^- \slashed \partial \psi^- - \bar{\psi}^- \hat{\omega} \psi^+ + (+ \leftrightarrow -).
\end{equation*}
Since $\gamma^5$ eigenvalues correspond to electric charge on an internal spacetime, the two mass terms $\bar{\psi}^{\mp} \hat{\omega} \psi^{\pm}$ represent
\begin{itemize}
 \item[(\textsc{a})] couplings between pointons of opposite charge; and
 \item[(\textsc{b})] pair creation/annihilation of pointons of opposite charge.
\end{itemize}
We will find that these two new interpretations allow us to construct a composite model of the standard model particles.

\vspace*{.25cm}

\textbf{(\textsc{a}) Couplings between pointons of opposite charge.}

Suppose $\psi^-$ and $\psi^+$ are pointon spinors of opposite charge, equal angular frequency $\omega$, and coincident worldlines $\beta$ in $\tilde{M}$.
Excitations of the mass terms $\bar{\psi}^{\pm} \hat{\omega} \psi^{\mp}$ of the Dirac Lagrangian $\mathcal{L} = \bar{\psi}\delta \psi$ then represent a coupling between the two pointons. 
The equations of motion are obtained by varying the respective fields $\bar{\psi}^{\pm}$ (or equivalently, $\psi^{\pm}$):
\begin{equation} \label{apex equations}
i \slashed \partial \psi^+ = \hat{\omega} \psi^- \ \ \ \ \text{ and } \ \ \ \ i \slashed \partial \psi^- = \hat{\omega} \psi^+.
\end{equation}
These equations hold at points $\beta(t) \in \tilde{M}$ where the spinor particles of the two pointons intersect in $M_{\beta(t)}$.

Denote by $[\psi^-]$ (resp.\ $[\psi^+]$) the upper (resp.\ lower) two components of $\psi^-$ (resp.\ $\psi^+$).
In the pointon's rest frame, the equations of motion (\ref{apex equations}) simplify to
\begin{equation*} \label{ngu2}
\tfrac 1c \partial_0 [\psi^+] = \pm \tfrac ir  [\psi^-] \ \ \ \ \text{ and } \ \ \ \ \tfrac 1c \partial_0 [\psi^-] = \pm \tfrac ir [\psi^+].
\end{equation*}
These are remarkably similar to the Faraday and (source-free) Amp\'ere equations,
\begin{equation*}
\tfrac 1c \partial_0 \bm{B} = - \nabla \! \times \! \bm{E} \ \ \ \ \text{ and } \ \ \ \ \tfrac 1c \partial_0 \bm{E} = \nabla \! \times \! \bm{B},
\end{equation*}
\textit{since $[\psi^-]$ and $[\psi^+]$ specify the rotation of a spinor particle.}
Just as Maxwell's equations describe the coupling of the electric and magnetic fields, \textit{the Dirac equation 
\begin{equation} \label{coupling equation}
\delta(\psi^- + \psi^+)(p) = 0
\end{equation}
describes the coupling of pointons of opposite charge.}
A solution to the Dirac equation on internal spacetime therefore represents a bound state of two pointons.

\begin{Remark} \label{Dirac bosons} \rm{
It was shown in \cite{B6} that pointons have spin $\tfrac 12$, and thus are fermions.
Furthermore, a bound state of an even number of fermions, such as a Cooper pair, is a boson.
Therefore \textit{solutions to the Dirac equation are bosons.} 
We thus arrive at an interpretation of the Dirac equation that is opposite to that of standard quantum theory: solutions to the Dirac equation are fermions in quantum theory, whereas solutions are bosons in the framework of internal spacetime. 
} \end{Remark}

Recall that the spinors ${[a0]}$ and ${[0b]}$, $a,b \in \{ \uparrow, \downarrow \}$, represent pointons with negative and positive charges respectively, and set ${[ab] := [a0] + [0b]}$.
Explicitly, 
\begin{center}
\begin{tabular}{lcl}
$[\uparrow \downarrow ] := [\uparrow \! 0] + [ 0 \! \downarrow] = \left[ \begin{smallmatrix} e^{i \omega t} \\ 0 \\ -e^{i \omega t} \\ 0 \end{smallmatrix} \right]_{\mathscr{C}}$
& \ \ \ \ & 
$[\downarrow \uparrow ] := [\downarrow \! 0 ] + [ 0 \! \uparrow] = \left[ \begin{smallmatrix} e^{-i \omega t} \\ 0 \\ -e^{-i \omega t} \\ 0 \end{smallmatrix} \right]_{\mathscr{C}}$
\\
$[\uparrow \uparrow ] := [\uparrow \! 0] + [ 0 \! \uparrow] = \left[ \begin{smallmatrix} e^{i \omega t} \\ 0 \\ -e^{-i \omega t} \\ 0 \end{smallmatrix} \right]_{\mathscr{C}}$
& \ \ \ \ & 
$[\downarrow \downarrow ] := [\downarrow \! 0 ] + [ 0 \! \downarrow] = \left[ \begin{smallmatrix} e^{-i \omega t} \\ 0 \\ -e^{i \omega t} \\ 0 \end{smallmatrix} \right]_{\mathscr{C}}$
\end{tabular}
\end{center}

\begin{Lemma} \label{[ab]}
For $a,b \in \{ \uparrow, \downarrow \}$, the pointon spinors ${[a0]}$ and ${[0b]}$ may couple.
\end{Lemma}

\begin{proof}
We must show that ${[a0]}$ and ${[0b]}$ satisfy (\ref{coupling equation}).
Indeed, we have
\begin{multline} \label{ttngu}
\delta([ \uparrow \! 0] + [0 \! \uparrow ]) = (i \slashed \partial - \hat{\omega}) [ \uparrow \uparrow ] \\
 = \left[ \begin{matrix} 
-\hat{\omega} & i \partial_0 \\ 
i \partial_0 & -\hat{\omega} \end{matrix} \right] 
\left[ \begin{smallmatrix} e^{i \omega t} \\ 0 \\ -e^{-i \omega t} \\ 0 \end{smallmatrix} \right]_{\mathscr{C}}  
= \left[ \begin{smallmatrix} -\omega e^{i \omega t} - \omega e^{-i \omega t} \\ 0 \\ \omega e^{i \omega t} + \omega e^{-i \omega t} \\ 0 \end{smallmatrix} \right]_{\mathscr{C}}
= 2 \omega \left[ \begin{smallmatrix} -\cos (\omega t) \\ 0 \\ \cos (\omega t) \\ 0 \end{smallmatrix} \right]_{\mathscr{C}}.
\end{multline}
Thus, $\delta [ \uparrow \uparrow ](t) = -\delta [\downarrow \downarrow] (t) = 0$ for $t = \tfrac{(2n+1) \pi }{2\omega}$, $n \in \mathbb{Z}$.
Similarly, 
\begin{multline} \label{ttngu2}
\delta([ \uparrow \! 0] + [0 \! \downarrow ]) = (i \slashed \partial - \hat{\omega}) [ \uparrow \downarrow ] \\
 = \left[ \begin{matrix} 
-\hat{\omega} & i \partial_0 \\ 
i \partial_0 & -\hat{\omega} \end{matrix} \right] 
\left[ \begin{smallmatrix} e^{i \omega t} \\ 0 \\ -e^{i \omega t} \\ 0 \end{smallmatrix} \right]_{\mathscr{C}}  
= \left[ \begin{smallmatrix} -\omega e^{i \omega t} + \omega e^{i \omega t} \\ 0 \\ -\omega e^{i \omega t} + \omega e^{i \omega t} \\ 0 \end{smallmatrix} \right]_{\mathscr{C}}
= \left[ \begin{smallmatrix} 0 \\ 0 \\ 0 \\ 0 \end{smallmatrix} \right].
\end{multline}
\end{proof}

\begin{Remark} \rm{
It is essential in (\ref{ttngu}) that we use the operator $\hat{\omega}$, and not $\pm \omega$, since 
\begin{equation*}
(i \slashed \partial + \omega) [ \uparrow \uparrow ] = 2\omega \left[ \begin{smallmatrix} \cos (\omega t) \\ 0 \\ - \sin (\omega t) \\ 0 \end{smallmatrix} \right]
\ \ \ \ \text{ and } \ \ \ \ 
(i \slashed \partial - \omega) [ \uparrow \uparrow ] = -2\omega \left[ \begin{smallmatrix} \sin (\omega t) \\ 0 \\ \cos (\omega t) \\ 0 \end{smallmatrix} \right]
\end{equation*}
are both nonvanishing for all $t$.
}\end{Remark}

In Proposition \ref{ngu} we found that the spinors ${[a0]}$ and ${[0a]}$ represent single pointons using the circular trajectory (\ref{spinon}), and in the following we again show that these spinors represent single pointons, but instead by using the Dirac equation as described in Remark \ref{Dirac bosons}.
This shows consistency within our model.

\begin{Proposition} \label{ngu3} \
\begin{enumerate}
\item[(i)] Each of the spinors $[a0]$, $[0a]$, with $a \in \{ \uparrow, \downarrow \}$, is not in the kernel of $\delta$ for any $t$, and so represents a single pointon. 
Thus, each spinor is a fermion. 
\item[(ii)] Each of the spinors $[ab]$, with $a,b \in \{ \uparrow, \downarrow \}$, is in the kernel of $\delta$ for periodic $t$, and so represents a bound state of two pointons.
Thus, each spinor is a neutral boson.
\end{enumerate}
\end{Proposition}

\begin{proof}
(i) is straightforward to verify. 
(ii) holds by (\ref{ttngu}) and (\ref{ttngu2}).
\end{proof}

Two spinors $\psi_j = \psi_j^- + \psi_j^+$, $j \in \{ 1,2 \}$, may couple via the mass term $\bar{\psi}_1 \hat{\omega} \psi_2$ of $\mathcal{L} = \bar{\psi}_1 \delta \psi_2$, yielding a bound state of at most four pointons: two of positive charge, $\psi_1^+$, $\psi_2^+$, and two of negative charge, $\psi_1^-$, $\psi_2^-$.
We call such a state a \textit{geom}, for `geometric atom', and call $\psi_1$ and $\psi_2$ its \textit{orbitals}. 
We denote a geom $\bar{\psi}_1 \hat{\omega} \psi_2$ by its unordered pair, $[\psi_1, \psi_2] = [\psi^-_1 \psi^+_1, \psi^-_2 \psi^+_2] = [\psi^-_2 \psi^+_2, \psi^-_1 \psi^+_1] = [\psi_2, \psi_1]$. 

Since pointons are fermions, geoms are constrained by the Pauli exclusion principle: a geom cannot contain two pointons of equal charge and spin.
For example, the states ${[\uparrow \! 0, \downarrow \downarrow ]}$ and ${[\uparrow \downarrow, \downarrow \uparrow ]}$ are allowed, whereas ${[\uparrow \! 0, \uparrow \downarrow ]}$ and ${[\uparrow \downarrow, \downarrow \downarrow ]}$ are not allowed since they violate the Pauli exclusion principle.

\begin{Remark} \label{observations} \rm{
We make two observations:
\begin{itemize}
 \item A geom is a boson (resp.\ fermion) if it contains an even (resp.\ odd) number of pointons, since pointons are spin-$\tfrac 12$ fermions.
 \item The electric charge of a geom is the sum of the electric charges of its pointons.
\end{itemize}
}\end{Remark}

\vspace*{.25cm}

\textbf{(\textsc{b}) Pair creation/annihilation of pointons of opposite charge.} 

Excitations of the mass terms $\bar{\psi}^{\pm} \hat{\omega} \psi^{\mp}$ of $\mathcal{L} = \bar{\psi}\delta \psi$ may also represent pair creation/annihilation of two pointons of opposite charge, that is, a two-valent interaction vertex. 
We thus say that at $p \in \tilde{M}$, two pointon spinors $\psi_1$, $\psi_2$ may couple, or two orbitals $\psi_1$, $\psi_2$ may \textit{fuse}, if whenever $\psi_1$, $\psi_2$ are charged resp.\ neutral we have
\begin{equation} \label{interaction equation}
\delta(\psi_1^{\pm} + \psi_2^{\mp})(p) = 0 \ \ \ \text{ resp.\ } \ \ \ \delta(\psi_1 + \psi_2)(p) = 0, 
\end{equation}
and 
\begin{equation} \label{can't fuse}
\delta \psi_1 \equiv 0 \ \ \Longleftrightarrow \ \ \delta \psi_2 \equiv 0.
\end{equation}

\begin{Lemma}
For $a,b \in \{ \uparrow, \downarrow \}$, the pointon spinors ${[a0]}$ and ${[0b]}$ may fuse. 
\end{Lemma}

\begin{proof}
See the proof of Lemma \ref{[ab]}.
\end{proof}

In general, two spinors $\psi_1 = {[ab]}$ and $\psi_2 = {[cd]}$, $a,b,c,d \in \{ 0, \uparrow, \downarrow \}$, that satisfy (\ref{interaction equation}) and (\ref{can't fuse}) may fuse into one of three spinors ${[ad]}$, ${[cb]}$, ${[00]}$:
\begin{equation*}
\text{\footnotesize{$
{\arraycolsep=2.9pt 
\begin{array}{ccc}
\xymatrix @C=-.3pc @R-2pc{
[ a & b \ar@{}[ld]^(.11){}="e"^(.95){}="f" \ar@{-} "e";"f" ]\\
[ c & d ] \\
\ar@{}[rr]^(-.45){}="a"^(1.0){}="b" \ar@{-} "a";"b" &&&&\\
[ a & d ] }
&
\xymatrix @C=-.3pc @R-2pc{
[ a \ar@{}[rd]^(.05){}="e"^(.9){}="f" \ar@{-} "e";"f" & b ]\\
[ c & d ] \\
\ar@{}[rr]^(-.45){}="a"^(1.0){}="b" \ar@{-} "a";"b" &&&&\\
[ c & b ] }
&
\xymatrix @C=-.3pc @R-2pc{
[ a \ar@{}[rd]^(.05){}="e"^(.9){}="f" \ar@{-} "e";"f" & b \ar@{}[ld]^(.11){}="e"^(.95){}="f" \ar@{-} "e";"f" ]\\
[ c & d ] \\
\ar@{}[rr]^(-.45){}="a"^(1.0){}="b" \ar@{-} "a";"b" &&&&\\
[ 0 & 0 ] }
\end{array} }
$}}
\end{equation*}
where the diagonal lines indicate pointon pair annihilation $\bar{\psi}_1^{\pm} \hat{\omega} \psi_2^{\mp}$.

By writing the spinors ${[ab]}$, $a, b \in \{ \uparrow, \downarrow \}$, in the Dirac basis using (\ref{chiral - Dirac}), we see that the spinors ${[\uparrow \downarrow]}$ and ${[\downarrow \uparrow]}$ are $\gamma^0$ eigenspinors, whereas ${[\uparrow \uparrow]}$ and ${[\downarrow \downarrow]}$ are not.
Moreover, similar to the pointon spinors, ${[\uparrow \downarrow]}$ and ${[\downarrow \uparrow]}$ may be obtained from the circular trajectory (\ref{spinon}) with the roles of $\gamma(e_0) = \gamma^0$ and $\gamma^5$ swapped:
\begin{equation*}
[\uparrow \downarrow] = \left[ \begin{smallmatrix} e^{i \omega t} \\ 0 \\ -e^{i \omega t} \\ 0 \end{smallmatrix} \right]_{\mathscr{C}} = \left[ \begin{smallmatrix} 0 \\ 0 \\ -e^{i \omega t} \\ 0 \end{smallmatrix} \right]_{\mathscr{D}} \ \ \ \ \ \text{ and } \ \ \ \ \
[\downarrow \uparrow] = \left[ \begin{smallmatrix} e^{-i \omega t} \\ 0 \\ -e^{-i \omega t} \\ 0 \end{smallmatrix} \right]_{\mathscr{C}} = \left[ \begin{smallmatrix} 0 \\ 0 \\ -e^{-i \omega t} \\ 0 \end{smallmatrix} \right]_{\mathscr{D}}.
\end{equation*}
\textit{It is therefore natural to extend the generating set of our composite model to include a maximal set of covariantly independent pointon $\gamma^0$ eigenspinors, just as the generators include a maximal set of covariantly independent pointon $\gamma^5$ eigenspinors.}

\begin{Proposition} \label{ngu2!} \
The spinors
\begin{equation} \label{max set2}
\left[ \begin{smallmatrix} e^{i \omega t} \\ 0 \\ 0 \\ 0 \end{smallmatrix} \right]_{\mathscr{D}} =: [ \, \downarrow \, ], \ \ \ 
\left[ \begin{smallmatrix} e^{-i \omega t} \\ 0 \\ 0 \\ 0 \end{smallmatrix} \right]_{\mathscr{D}} =: [ \, \uparrow \, ], \ \ \ 
\left[ \begin{smallmatrix} 0 \\ 0 \\ -e^{-i \omega t} \\ 0 \end{smallmatrix} \right]_{\mathscr{D}} = [ \downarrow \uparrow ], \ \ \ 
\left[ \begin{smallmatrix} 0 \\ 0 \\ -e^{i \omega t} \\ 0 \end{smallmatrix} \right]_{\mathscr{D}} = [ \uparrow \downarrow ]
\end{equation}
form a maximal covariantly independent set of pointon $\gamma^0$ eigenspinors, denoted $\mathcal{O}_0$.
\end{Proposition}

\begin{proof}
Similar to the proof of Proposition \ref{ngu}.
\end{proof}

Using $u = \omega r = \omega m^{-1}$, the $\gamma^0$ eigenspinors (\ref{max set2}) decompose in the chiral basis as
\begin{equation*}
\begin{array}{lcl}
\text{{\tiny $\sqrt{2}u$}} [\, \downarrow \,  ] = \text{{\tiny $u$}} \! \left[ \begin{smallmatrix} e^{i \omega t} \\ 0 \\ e^{i \omega t} \\ 0 \end{smallmatrix} \right]_{\mathscr{C}} \!
 = \tfrac{\omega}{m}[\uparrow \! 0] + \tfrac{\omega}{-m} [0 \! \downarrow]
& \ \ \ &
\text{{\tiny $\sqrt{2}u$}} [\, \uparrow \, ] = \text{{\tiny $u$}} \! \left[ \begin{smallmatrix} e^{-i \omega t} \\ 0 \\ e^{-i \omega t} \\ 0 \end{smallmatrix} \right]_{\mathscr{C}} \!
= \tfrac{\omega}{m}[\downarrow \! 0] + \tfrac{\omega}{-m}[0 \! \uparrow]
\\
\text{{\tiny $\sqrt{2}u$}} [ \downarrow \uparrow ] = \text{{\tiny $u$}} \! \left[ \begin{smallmatrix} e^{-i \omega t} \\ 0 \\ -e^{-i \omega t} \\ 0 \end{smallmatrix} \right]_{\mathscr{C}}  \! 
= \tfrac{\omega}{m}[\downarrow \! 0 ] + \tfrac{\omega}{m}[ 0 \! \uparrow] 
& &
\text{{\tiny $\sqrt{2}u$}} [ \uparrow \downarrow ]  = \text{{\tiny $u$}} \! \left[ \begin{smallmatrix} e^{i \omega t} \\ 0 \\ -e^{i \omega t} \\ 0 \end{smallmatrix} \right]_{\mathscr{C}} \! 
= \tfrac{\omega}{m}[ \uparrow \! 0 ] + \tfrac{\omega}{m}[ 0 \! \downarrow]
\end{array}
\end{equation*}
In contrast to ${[\uparrow \downarrow]}$ and ${[\downarrow \uparrow]}$, the spinors ${[\, \downarrow \, ]}$ and ${[\, \uparrow \, ]}$ cannot decay into their charged $\gamma^5$ eigenspinor summands because $\omega$ is positive by definition, and therefore one of the two summands has negative mass, $-m$. 
Indeed, in the following we show that these spinors are \textit{not} bound states of two pointons using the Dirac equation, thus again showing consistency within our model.

\begin{Lemma} \label{parity violation cor}
The spinors ${[\, \downarrow \, ]}$ and ${[\, \uparrow \, ]}$ each represent a single pointon.
Furthermore, they admit a single nonvanishing fusion.
\end{Lemma}

\begin{proof}
First observe that
\begin{equation} \label{co}
\delta [\, \downarrow \, ] = (i \slashed \partial - \hat{\omega}) [\, \downarrow \,  ]
= \left[ \begin{matrix} 
-\hat{\omega} & i \partial_0 \\ 
i \partial_0 & -\hat{\omega} \end{matrix} \right] 
\left[ \begin{smallmatrix} e^{i \omega t} \\ 0 \\ e^{i \omega t} \\ 0 \end{smallmatrix} \right]_{\mathscr{C}} = -2 \omega \left[ \begin{smallmatrix} e^{i \omega t} \\ 0 \\ e^{i \omega t} \\ 0 \end{smallmatrix} \right]_{\mathscr{C}}
\end{equation}
is nonvanishing for all $t$.
Similarly, $\delta {[ \, \uparrow \, ]} (t) \not = 0$ for all $t$. 
Thus, ${[\, \downarrow \, ]}$ and ${[\, \uparrow \, ]}$ each represent a single pointon.
The spinors may fuse since
\begin{multline} \label{updown}
\delta([\, \downarrow \, ] + [\, \uparrow \, ]) = (i \slashed \partial - \hat{\omega}) ([\, \downarrow \, ] + [\, \uparrow \, ]) \\
 = \left[ \begin{matrix} 
-\hat{\omega} & i \partial_0 \\ 
i \partial_0 & -\hat{\omega} \end{matrix} \right] 
\left( \left[ \begin{smallmatrix} e^{i \omega t} \\ 0 \\ e^{i \omega t} \\ 0 \end{smallmatrix} \right]_{\mathscr{C}}  + \left[ \begin{smallmatrix} e^{-i \omega t} \\ 0 \\ e^{-i \omega t} \\ 0 \end{smallmatrix} \right]_{\mathscr{C}} \right)
= -2 \omega \left[ \begin{smallmatrix} e^{i \omega t} - e^{-i \omega t} \\ 0 \\ e^{i \omega t} - e^{-i \omega t} \\ 0 \end{smallmatrix} \right]_{\mathscr{C}}
= -4 \omega \left[ \begin{smallmatrix} \sin (\omega t) \\ 0 \\ \sin (\omega t) \\ 0 \end{smallmatrix} \right]_{\mathscr{C}}
\end{multline}
periodically vanishes for $t = \tfrac{n \pi }{\omega}$, $n \in \mathbb{Z}$. 
However, since each spinor represents a single pointon, they can only admit one nonvanishing fusion (just as the pair ${[\downarrow \! 0]}$ and ${[0 \! \downarrow]}$ only admits one nonvanishing fusion).
\end{proof}

In \cite{B1} we will derive electroweak parity violation from this lemma.

Let $\beta \subset \tilde{M}$ be a timelike wordline with $4$-velocity $v = e_0$ and spinor field ${[\, \downarrow \, ]}$ or ${[\, \uparrow \, ]}$.
Since  ${[\, \downarrow \, ]}$ and ${[\, \uparrow \, ]}$ represent neutral particles, the timelike direction $e_0$ does not vanish along $\beta$.
Thus, in the absense of other particles, the dimension of the internal tangent spaces $M_{\beta(t)} \subseteq \tilde{M}_{\beta(t)}$ along $\beta$ is four, $\dim M_{\beta(t)} = 4$.
In particular, $M_{\beta(t)} = \tilde{M}_{\beta(t)}$.
\textit{The pointon spinors ${[\, \downarrow \, ]}$, ${[\, \uparrow \, ]}$ therefore represent spinor particles $\alpha \subset \tilde{M}$, given in (\ref{spinon}), whose central worldlines $\beta$ are not pointons.}\footnote{We nevertheless continue to call these spinors `pointon spinors' for ease of notation.}$^,$\footnote{In our composite model (see Section \ref{composite section}), electron neutrinos $\nu_e$ and anti-neutrinos $\bar{\nu}_e$ are identified with the respective geoms ${[ \, \uparrow \, , 00]}$ and ${[ \, \downarrow \, , 00]}$.
It is well known that proton fusion reactions in the sun yield an electron neutrino flux of approximately $6 \times 10^{10} \, \text{cm}^{-2}\text{s}^{-1}$ on the surface of the Earth.
If the dimensions of the tangent spaces $M_{\beta(t)}$ along neutrino geom worldlines $\beta \subset \tilde{M}$ were less than four, then it is plausible that the enormous flux of (extremely weakly interacting) neutrinos would continually collapse the spin vectors of electrons on the Earth, contrary to observation.
Fortunately, neutrinos are identified with geoms whose tangent spaces are four dimensional, and therefore an electron's spin vector is not affected by passing neutrinos.} 
That is, they are spinors which have broken free from their central pointon worldlines, and thus propagate freely without being bound to a pointon.
However, if the spinors ${[\, \downarrow \, ]}$, ${[\, \uparrow \, ]}$ have color charge (see Section \ref{color section}), then their central worldlines $\beta$ will necessarily be pointons; this is addressed in \cite{B2}.

A fusion of ${[\, \downarrow \, ]}$ and ${[\, \uparrow \, ]}$ has total angular momentum zero and is neutral.
We therefore take their nonvanishing fusion to be ${[\uparrow \uparrow]}$ or ${[\downarrow \downarrow]}$.
(Recall that the spinor particles of ${[a0]}$ and ${[0a]}$ rotate in opposite directions, so ${[\uparrow \uparrow]}$ and ${[\downarrow \downarrow]}$ have total angular momentum zero.)
Furthermore, it can only be one of ${[\uparrow \uparrow]}$ or ${[\downarrow \downarrow]}$, by Lemma \ref{parity violation cor}.
We therefore choose, once and for all, one of these spinors as their nonvanishing fusion; let us choose ${[\downarrow \downarrow]}$.
For clarity, then, we rewrite ${[\, \downarrow \, ]}$ and ${[\, \uparrow \, ]}$ as ${[\downarrow \! *]}$ and ${[ * \! \downarrow]}$:
\begin{equation*} \label{star apex}
\text{\footnotesize{$
{\arraycolsep=2.9pt 
\begin{array}{rl}
\xymatrix @C=-.3pc @R-2pc{
[ \ \downarrow \  ] \ \ \ = \\
[ \ \uparrow \  ] \ \ \ =}
&
\xymatrix @C=-.3pc @R-2pc{
[ \downarrow & {*} \ar@{}[ld]^(.11){}="e"^(.95){}="f" \ar@{-} "e";"f" ]\\
[ {*} & \downarrow ] \\
\ar@{}[rr]^(-.45){}="a"^(1.0){}="b" \ar@{-} "a";"b" &&&&\\
[ \downarrow & \downarrow ] }
\end{array} }
$}}
\end{equation*}

The pointon spinors with negative (resp.\ positive; zero) charge are thus ${[a0]}$ (resp.\ ${[0a]}$; ${[\downarrow \! *]}$, ${[* \! \downarrow ]}$), with $a \in \{ \uparrow, \downarrow \}$, and the set of all orbitals is $\mathcal{O}_0 \cup \mathcal{O}_5 \cup \{ {[\uparrow \uparrow]}, {[\downarrow \downarrow ]} \}$, where 
\begin{equation*}
\mathcal{O}_5 := \{ [\uparrow \! 0], [\downarrow \! 0], [0 \! \uparrow ], [0 \! \downarrow] \} 
\ \ \ \ \text{ and } \ \ \ \ 
\mathcal{O}_0 := \{ [ \uparrow \downarrow], [\downarrow \uparrow ], [\downarrow \! *], [ * \! \downarrow] \}.
\end{equation*}

\begin{Theorem}
The pairs of orbitals that may fuse are precisely
\begin{equation*} \label{fusion pairs}
\{ [a0], [0b] \}, \ \ \ \ \{ [ab], [ab] \}, \ \ \ \ \{ [ab], [ba] \}, \ \ \ \ \{ [aa], [bb] \}, \ \ \ \ \{ [\downarrow \! *], [* \! \downarrow] \},
\end{equation*}
with $a,b \in \{ \uparrow, \downarrow \}$. 
\end{Theorem}

\begin{proof}
The allowed fusions are straightforward to verify.
For example, ${[\downarrow \! {*}]}$ and ${[\downarrow \! {*}]}$ cannot fuse since $\delta({[\downarrow \! {*}]} + {[\downarrow \! {*}]} )(t) \not = 0$ for all $t$ by (\ref{co}).
Furthermore, the orbitals ${[aa]}$ and ${[bc]}$, with $a,b,c \in \{ \uparrow, \downarrow \}$, $b \not = c$, cannot fuse by (\ref{can't fuse}), since $\delta {[ bc]} \equiv 0$ by (\ref{ttngu2}), whereas $\delta {[aa]} \not \equiv 0$ by (\ref{ttngu}).
\end{proof}

This theorem will play an essential role in determining interaction vertices in \cite{B1}.

\section{Color charge} \label{color section}

Let $\beta \subset \tilde{M}$ be the worldline of a timelike pointon with $4$-velocity $v$, and let $e_0 = v, e_1, e_2, e_3$ be an orthonormal tetrad along $\beta$.
An essential feature of internal spacetime geometry has been the identification of a free orientation $o_0 \in \{ \pm 1 \}$ of $e_0$ with the electric charge $e^{\pm}$ of the pointon.
We extend this identification to spacelike orientations: for $i \in \{ 1,2,3\}$, we identify a free orientation $o_i \in \{ \pm 1 \}$ of $e_i$ with a \textit{color charge} $c_i^{o_i}$.
There are therefore three color charges (we may call $c_1^-$, $c_2^-$, $c_3^-$ red, green, blue; and $c_1^+$, $c_2^+$, $c_3^+$ anti-red, anti-green, anti-blue).

Fix a (nonphysical) choice of orientation $o_{\tilde{M}} \in \{ \pm 1 \}$ of $\tilde{M}_{\beta(t)}$, say $o_{\tilde{M}} = 1$.
Then, noting that $o_{\tilde{M}} = o_{0123} = o_0o_1o_2o_3$, we have 
\begin{equation*}
o_1 = o_{\tilde{M}} o_1 = o_{\tilde{M}} o_1^{-1} = o_0o_2o_3 = o_{023}.
\end{equation*}
Similarly, $o_2 = o_{013}$ and $o_3 = o_{012}$. 
Recall that the electric charge $o_0$ of a pointon arises because it has extent in the $e_0$ direction in $\tilde{M}$, that is, its $4$-velocity $v$ is $v = e_0$.
Correspondingly, a pointon has color charge $o_i = o_{0jk}$, with $i,j,k$ distinct, if it has extent along a $(2+1)$-dimensional hypersurface with spatial tangent space $e_j \wedge e_k$. 
The resulting hypersurface would then be a single point in spacetime $M$. 
To summarize, \textit{we identify the free orientation of a vanishing $2$-dimensional subspace with spin, and that of a vanishing $1$- or $3$-dimensional subspace with a charge;} see Table \ref{charge and spin}.

\begin{table}
\label{charge and spin}
\caption{The pointon properties that emerge from vanishing subspaces of spacetime tangent spaces.}
\begin{center}
\begin{tabular}{|l|l|l|}
\hline
vanishing subspace & free orientation & identification\\
\hline \hline
$e_0$ & $o_0 = o_{123}$ & electric charge $e^{o_0}$\\
\hdashline
$e_0 \wedge e_2 \wedge e_3$ & $o_1 = o_{023}$ & color charge $r^{o_1}$\\
$e_0 \wedge e_3 \wedge e_1$ & $o_2 = o_{013}$ & color charge $g^{o_2}$\\
$e_0 \wedge e_1 \wedge e_2$ & $o_3 = o_{012}$ & color charge $b^{o_3}$\\
\hdashline
$e_2 \wedge e_3$ & $o_{23} = o_{01}$ & spin in the direction $o_{01}e_1$\\
$e_3 \wedge e_1$ & $o_{13} = o_{02}$ & spin in the direction $o_{02}e_2$\\
$e_1 \wedge e_2$ & $o_{12} = o_{03}$ & spin in the direction $o_{03}e_3$\\
\hline
\end{tabular}
\end{center}
\end{table}

We develop a geometric model of color charge based on these identifications in \cite{B2}.
From this model we find that
\begin{itemize}
 \item[(\textsc{i})] only orbitals with a single pointon may have nonzero color charge; and
 \item[(\textsc{ii})] a geom can have at most one unit of color charge, just as with electric charge. 
\end{itemize}
Thus, the only orbitals that may possess color charge are ${[a0]}$, ${[0a]}$, with $a \in \{ \uparrow, \downarrow \}$, and ${[\downarrow \! *]}$, ${[* \! \downarrow ]}$.
We denote these colored orbitals by
\begin{equation*}
\textcolor{red}{\pmb{(}} a0], \ \ \ \ [0 a \textcolor{red}{\pmb{)}}, \ \ \ \ [ \, \downarrow \textcolor{red}{\pmb{)}} = [\downarrow \! * \textcolor{red}{\pmb{)}}, \ \ \ \ \textcolor{red}{\pmb{(}} \uparrow \,  ] = \textcolor{red}{\pmb{(}} * \! \downarrow ],
\end{equation*}
where a colored left (resp.\ right) rounded bracket denotes color charge $c^-_i$ (resp.\ $c^+_i$).
We show in \cite{B2} that for $*$-orbitals, the rounded bracket must be adjacent to the $*$ component. 
Furthermore, colored $*$-orbitals are spin states of a \textit{charged} spin-$\tfrac 12$ particle, and thus there are four such states (positive/negative charge and spin up/down). 
This can only be obtained, then, by allowing $\uparrow$ arrows in colored $*$-orbitals,
\begin{equation*} \label{four}
[ \, \downarrow \textcolor{red}{\pmb{)}} = [\downarrow \! * \textcolor{red}{\pmb{)}}, \ \ \ \ 
\textcolor{red}{\pmb{(}} \downarrow \,  ] = \textcolor{red}{\pmb{(}} * \! \uparrow], \ \ \ \ 
\textcolor{red}{\pmb{(}} \uparrow \,  ] = \textcolor{red}{\pmb{(}} * \! \downarrow], \ \ \ \ 
[ \, \uparrow \textcolor{red}{\pmb{)}} = [\uparrow \! * \textcolor{red}{\pmb{)}}.
\end{equation*}

We denote a geom with total color charge $c^-_i$ (resp.\ $c^+_i$) by a colored left (resp.\ right) rounded bracket, just as for orbitals.
The set of all geoms with color is given in Table \ref{table1}.
The electric charge of a geom with color is obtained by first substituting $\pm \tfrac 13$ for each colored geom component with electric charge $\pm 1$, 
\begin{equation} \label{color electric}
c_i^{o_i} \mapsto \tfrac 13 e^{o_i},
\end{equation}
and then summing the electric charges of each component, just as in Remark \ref{observations}.

\section{A composite model of the standard model particles} \label{composite section}

Remarkably, the set of all possible geoms reproduces precisely three generations of leptons and quarks, the electroweak gauge bosons, the Higgs boson, and five new neutral bosonic polarization states,
\begin{equation} \label{new polarizations}
x_0 = {[\downarrow \! *, * \! \downarrow]}, \ \ \ x_{ab} = {[a 0, 0 b]},
\end{equation}
with $a,b \in \{ \uparrow, \downarrow \}$. 
These new polarization states partition into one of the following:
\begin{itemize}
 \item[(i)] one massive spin-$2$ boson with polarization basis $\{ x_0, x_{\uparrow \downarrow}, x_{\downarrow \uparrow}, x_{\uparrow \uparrow}, x_{\downarrow \downarrow} \}$;
 \item[(ii)] one massive spin-$1$ boson and one massless spin-$1$ boson, with respective bases $\{ x_0, x_{\uparrow \downarrow}, x_{\downarrow \uparrow} \}$, $\{ x_{\uparrow \uparrow}, x_{\downarrow \downarrow} \}$; or
 \item[(iii)] one spin-$0$ boson and two massless spin-$1$ bosons, with respective bases $\{ x_0 \}$, $\{ x_{\uparrow \downarrow}, x_{\downarrow \uparrow} \}$, $\{ x_{\uparrow \uparrow}, x_{\downarrow \downarrow} \}$.
\end{itemize}
Based on a pattern of geom masses given in \cite{B1}, we expect that (i) holds.
Furthermore, in \cite{B2} we model gluons as extended objects similar to flux tubes or strings, using internal spacetime geometry.
In particular, gluons are not point particles in our model, and thus are not geoms.

\begin{Theorem}
The set of all geoms partition precisely into spin/polarization bases for each standard model fermion, electroweak boson, and the Higgs boson, together with the new bosonic polarization basis (\ref{new polarizations}).
Furthermore, each geom has the correct spin type (fermion or boson), electric charge, and color charge as the corresponding standard model particle. 
\end{Theorem}

\begin{proof}
The standard model geoms are given in Table \ref{table1}.
Their spin types and electric charges are obtained from Remark \ref{observations} and (\ref{color electric}), and their color charges are obtained from (\textsc{i}) and (\textsc{ii}) in Section \ref{color section}. 

The neutral pointon spinors $\psi_1 = {[ \, \downarrow \, ]} = {[\downarrow \! *]}$, $\psi_2 = {[ \, \uparrow \, ]} = {[* \! \downarrow]}$ may couple to form the geom $x_0$, by (\ref{updown}) and (\ref{interaction equation}).
In contrast, (\ref{interaction equation}) implies that the charged pointon spinors $\psi_1 = {[a0]}$, $\psi_2 = {[b0]}$ (resp.\ $\psi_1 = {[0a]}$, $\psi_2 = {[0b]}$), with $a,b \in \{ \uparrow, \downarrow \}$, cannot couple to form the geom ${[a0,b0]}$ (resp.\ ${[0a,0b]}$).
Indeed, we have $\psi_1^+ = \psi_2^+ = 0$ (resp.\ $\psi_1^- = \psi^-_2 = 0$).
Finally, $\psi_1 = {[a0]}$ and $\psi_2 = {[0b]}$ may couple by Lemma \ref{[ab]}, yielding either ${[a0,0b]}$ or their fusion ${[ab,00]}$.\footnote{Positronium is a bound state of an electron and positron, each with its own worldline.
The geoms ${[a0,0b]}$, ${[ab,00]}$ are in a sense also bound states of an electron and positron, but with a single shared wordline.}
\end{proof}

Geom interactions, including the interaction vertices of the predicted new boson, are investigated in \cite{B1}.

\begin{Remark} \rm{
The four polarization states $x_{ab}$, $a,b \in \{ \uparrow, \downarrow \}$, are excluded if we impose the condition that if a geom has two nonempty orbitals, then at least one must be `full' (in the sense that both of its components are nonzero).
The state $x_0$ is excluded if we impose the condition that a geom may contain at most one $*$-component.
Although the standard model geoms would be unaffected by imposing either condition, we currently have no reason to assume these extra conditions.
}\end{Remark}

\appendix \label{appendix}
\section{Mass shell assumptions in derivations of $E_0 = m$}

We consider three `derivations' of $E_0 = mc^2$, and show that no contradiction arises between these derivations 
and our generalization $E_0 = mcu$.

\subsection{By definition}

Consider a particle with $4$-momentum $k^a_0 = (E_0/c, 0, 0, 0)$ in its rest frame.
It is standard to define the particle's mass $m$ to be the norm of $k^a$ divided by $c$, that is, $k^2 = k^a k_a = (mc)^2$.
Consequently, $E_0/c = mc$.
We emphasize that this is not a \textit{derivation} of $E_0 = mc^2$, but simply a consequence of the \textit{definition} of $m$.

Indeed, we may instead define the particle's mass $m$ by $k^2 = (mu)^2$, where $u$ is some parameter intrinsic to the particle (in our case the velocity of the pointon's spinor particle).
It then follows that $E_0/c = mu$, whence $E_0 = mcu$.\footnote{In this case, the particle's kinetic energy is $\tfrac 12 (u/c) mv^2$ since
\begin{equation*}
E/c = (E_0/c) \frac{dt}{d\tau} = mu \frac{dt}{d\tau} = mu \left( 1 - \tfrac{v^2}{c^2} \right)^{-1/2} = mu + \tfrac 12 (u/c^2) mv^2 + O(v^4).
\end{equation*}}

\subsection{Derivation 1}

The following is based on \cite[p.\ 254]{TW}.
Consider a closed box of mass $M$ and width $L$ that is initially at rest.
An on-shell photon with energy $E_{\gamma}$ is emitted from, say, the left wall, travels to the right, and is then absorbed by the right wall.
By the conservation of $4$-momentum, the box must move with some speed $v$ to the left while the photon travels from the left side of the box to the right.
Let $\Delta x$ be the total distance the box is moved.
The total distance the photon travels between emission and absorption is then $L - \Delta x$. 
Since the photon is on shell, it travels with speed $c$, and thus the time from emission to absorption is $t = (L - \Delta x)/c$.
Let $m$ be the mass transported by the photon from the left wall to the right wall.

The $3$-momentum $Mv$ of the moving box is equal and opposite to the $3$-momentum $p_{\gamma}$ of the photon; whence
\begin{equation*}
Mv = -p_{\gamma}  = -E_{\gamma}/c,
\end{equation*}
where the second equality holds because the photon is on shell.
Therefore
\begin{equation*}
\Delta x = vt = -E_{\gamma}(Mc)^{-1} \cdot (L- \Delta x)c^{-1}.
\end{equation*}
But the center-of-mass of the photon + box system is stationary, and so
\begin{equation*}
M \Delta x + m(L - \Delta x) = 0.
\end{equation*}
Consequently, 
\begin{equation*}
m = -\frac{\Delta x M}{L - \Delta x} = -\frac{E (L- \Delta x)}{Mcu} \cdot \frac{M}{L - \Delta x} = \frac{E_{\gamma}}{c^2}.
\end{equation*}

This derivation is compatible with $E_{\gamma} = mcu$ since it assumes the photon is on shell, in which case $u = c$. 

\subsection{Derivation 2}

The following is due to Einstein \cite[p.\ 16]{E}, and is based on \cite[III.6]{Z}.
Consider a particle at rest that emits two on-shell photons of energy $E_{\gamma}$ in opposite directions.
Denote by $m_1$ and $E_1$ (resp.\ $m_2$ and $E_2$) the mass and rest energy of the particle before (resp.\ after) the emission. 
Let $\pm \theta$ be the angles the photons make with the $x$-axis in a frame moving with velocity $-v$ in the $x$-direction relative to the particle.
Then, by $3$-momentum conservation,
\begin{equation*}
m_1 v = m_2 v + 2 p_{\gamma, x} = m_2 v + 2 \cos \theta E_{\gamma}/c = m_2 v + 2 (v/c) E_{\gamma}/c,
\end{equation*}
where $p_{\gamma, x} = \cos \theta E_{\gamma}/c$ and $\cos \theta = v/c$ since the photons are on shell, whence the norms of their $4$-momenta and $4$-velocities vanish.
Thus, $E_{\gamma} = \tfrac 12 (m_1 - m_2) c^2$.
Furthermore, energy conservation in the rest frame of the particle implies
\begin{equation*}
E_1 = E_2 + 2 E_{\gamma} = E_2 + (m_1 - m_2)c^2.
\end{equation*}
Therefore,
\begin{equation*}
\Delta E = E_1 - E_2 = (m_1 - m_2)c^2 = \Delta m c^2.
\end{equation*}

This derivation is also compatible with $E_0 = mcu$ since it also assumes the photons are on shell.

\ \\
\textbf{Acknowledgments.}
The author thanks an anonymous referee for their careful reading and helpful comments.
The author was supported by the Austrian Science Fund (FWF) grant P 34854.

\bibliographystyle{hep}
\def\cprime{$'$} \def\cprime{$'$}

\end{document}